\documentclass[]{llncs}

\title{Completeness of Graphical Languages for Mixed States Quantum Mechanics} 

\titlerunning{Completeness of Graphical Languages for Mixed~States Quantum Mechanics}

\author{Titouan Carette\hfill Emmanuel Jeandel\hfill Simon Perdrix\hfill Renaud Vilmart
\institute{Universit\'e de Lorraine, CNRS, Inria, LORIA, F 54000 Nancy, France}
\\\email{\{firstname.name\}@loria.fr}
}

\authorrunning{T.Carette, E. Jeandel, S.Perdrix and R.Vilmart}

\usepackage{amssymb,nicefrac}
\usepackage{amsfonts}
\usepackage{mathtools}
\usepackage{calc}
\usepackage{multicol}
\usepackage{hyperref}
\usepackage{graphicx}
\usepackage{needspace}
\usepackage{nicefrac}

\usepackage{amsmath}

\usepackage{stmaryrd}
\usepackage{xspace}

\usepackage{geometry}
\usepackage{tikz}
\usepackage{pgfplots}
\usetikzlibrary{decorations.markings,arrows,decorations.pathmorphing}
\usetikzlibrary{matrix,backgrounds,folding,positioning}
\usetikzlibrary{shapes.geometric}
\pgfdeclarelayer{edgelayer}
\pgfdeclarelayer{nodelayer}
\pgfsetlayers{background,edgelayer,nodelayer,main}
\tikzstyle{every picture}=[baseline=-0.25em]
\tikzstyle{none}=[inner sep=0mm]
\tikzstyle{zxnode}=[shape=circle, minimum width=.25cm, inner sep=0.5pt, font=\footnotesize, draw=black]
\tikzstyle{gn}=[zxnode ,fill=green]
\tikzstyle{rn}=[zxnode ,fill=red]
\tikzstyle{bg}=[inner sep=0.5pt,minimum width=.25cm,fill=green,draw=white,very thick,shape=circle]
\tikzstyle{br}=[inner sep=0.5pt,minimum width=.25cm,fill=red,draw=white,very thick,shape=circle]
\tikzstyle{H box}=[rectangle,fill=yellow,draw=black,xscale=1,yscale=1,font=\footnotesize,inner sep=1.2pt,minimum width=0.15cm,minimum height=0.15cm]
\tikzstyle{ug}=[regular polygon, regular polygon sides=3, fill=red,draw=black,inner sep = 0pt,minimum width=1em]
\tikzstyle{black dot}=[inner sep=0.7mm,minimum width=0pt,minimum height=0pt,fill=black,draw=black,shape=circle]
\tikzstyle{dot}=[black dot]
\tikzstyle{white dot}=[dot,fill=white]
\tikzstyle{zwcross}=[diamond, draw, fill=gray, minimum width=0em, inner sep=1.5pt]
\tikzstyle{arrow}=[decoration={markings,mark=at position 1 with
	{\arrow[scale=1.2,>=stealth]{>}}},postaction={decorate}]

\tikzstyle{st}=[star,star points = 5, fill=white,draw=black,inner sep = 1.2pt,line width=1.2pt]
\tikzstyle{uglabel}=[rounded corners=0.2em,fill=green!20,inner sep=0.1em,font=\scriptsize, anchor=west, xshift=-0.2em, yshift=0,opacity=1]
\tikzstyle{braid}=[white, double=black, , line width=1.3pt, double distance=0.4 pt]

\tikzstyle{Box}=[rectangle,fill=white,draw=black,xscale=1,yscale=1,font=\footnotesize,inner sep=1.2pt,minimum width=0.4cm,minimum height=0.4cm]

\pgfdeclarelayer{edgelayer}
\pgfdeclarelayer{nodelayer}
\pgfsetlayers{background,edgelayer,nodelayer,main}
\tikzstyle{none}=[inner sep=0mm]
\tikzstyle{every loop}=[]

\newcommand{\tikzfig}[1]{
	\IfFileExists{./figures/#1.tikz}
	{\input{./figures/#1.tikz}}
	{\tikz[baseline=-0.5em]{\node[draw=red,font=\color{red},fill=red!10!white] {$#1$};}}
}

\def\fig{}

\newif\ifexternal
\externalfalse
\ifexternal

\usetikzlibrary{external}
\let\oldtikzfig\tikzfig
\renewcommand{\tikzfig}[1]{
	\tikzsetnextfilename{#1}
	\oldtikzfig{#1}}

\else

\fi

\usetikzlibrary{circuits.ee.IEC}
\newcommand{\ground}
{
\ifexternal\tikzsetnextfilename{ground}\fi
	\begin{tikzpicture}[circuit ee IEC,yscale=0.9,xscale=0.8]
	\draw[solid,arrows=-] (0,1ex) to (0,0) node[anchor=center,ground,rotate=-90,xshift=.66ex] {};
	\end{tikzpicture}}%
\newcommand{\bvdots}{ \tikz[baseline, every node/.style={inner sep=0}]{ \node at (0,0){.}; \node at (0,-3pt){.}; \node at (0,3pt){.}; } }
\newcommand{\sground}{\scalebox{0.5}{\!\ground}}

\newcommand{\cp  }{\sim_{\textnormal{cp}}}
\newcommand{\iso}{\sim_{\textnormal{iso}}}
\newcommand{\isop }{\sim_{\textnormal{iso}}^{+}}

\newcommand{\term}[1]{\text{ }!_{#1}}
\newcommand{\disc}[1]{\overset{#1}{\ground}}
\newcommand{\cterm}[1]{#1^{\term{}}}
\newcommand{\ciso}[1]{#1_{\textnormal{iso}}}
\newcommand{\cdisc}[1]{#1^{\sground}}

\def\rightangle{\tikz{\draw (0,0) -- (0,0.2) -- (0.2,0.2);}}

\newcommand{\eq}[2][~~]{
	#1
	\underset{\substack{#2}}{=}
	#1
}

\newcommand{\interp}[1]{\left\llbracket #1 \right\rrbracket}

\newcommand{\frag}[1]{$\frac{\pi}{#1}$-frag\-ment}
\newcommand{\titlerule}[1]{\noindent
	\begin{center}
		\rule{(\columnwidth-\widthof{#1})/2}{0.5pt}#1\rule{(\columnwidth-\widthof{#1})/2}{0.5pt}
\end{center}}
\newcommand{\annoted}[3]{{\scriptstyle #1}\left\lbrace\mathrlap{\phantom{#3}}\right.\overbrace{#3}^{#2}}

\newcommand{\bra}[1]{\ensuremath{\left\langle #1 \right|}}
\newcommand{\ket}[1]{\ensuremath{\left|  #1 \right\rangle}}

\newcommand{\tr}{\operatorname{tr}}
\def \zx {\textnormal{ZX}\xspace}

\allowdisplaybreaks

\begin{document}

\maketitle

\begin{abstract}	
	
	There exist several graphical languages for quantum information processing, like quantum circuits, ZX-Calculus, ZW-Calculus, etc. Each of these languages forms a $\dagger$-symmetric monoidal category ($\dagger$-SMC) and comes with an interpretation functor to the $\dagger$-SMC of (finite dimension) Hilbert spaces. In the recent years, one of the main achievements of the categorical approach to quantum mechanics has been to provide  
	several equational theories for most of these graphical languages, making them complete for various fragments of pure quantum mechanics. 
	
	We address the question of the extension of these languages beyond pure quantum mechanics, in order to reason on mixed states and general quantum operations, i.e.~completely positive maps. Intuitively, such an extension relies on the axiomatisation of a \emph{discard} map which allows one to get rid of a quantum system, operation which is not allowed in pure quantum mechanics. 
	
	We introduce a new construction, the \emph{discard construction}, which transforms any $\dagger$-symmetric monoidal category
	  into a symmetric monoidal category 
	  equipped with a discard map. Roughly speaking this construction consists in making any isometry causal. 
	
	Using this construction we provide an extension for several graphical languages that we prove to be complete for general quantum operations. However this construction fails for some fringe cases like the Clifford+T quantum mechanics, as the category does not have enough isometries.

%
%
%
%
%
%
%
%
%
%

\end{abstract}

\section{Introduction}



Graphical languages that speak of quantum information can be formalised through the notion of symmetric monoidal categories. Hence, it has a nice graphical representation using string diagrams \cite{selinger2010survey}. Qubits are represented by wires, and morphisms by graphical elements where some wires go in, and some others go out, just as in quantum circuits (which is actually a particular case of symmetric monoidal category), and where these graphical elements can be composed either in sequence (usual composition) or in parallel (tensor product). They usually come with an additional structure, a contravariant functor called dagger. 

Examples of graphical languages for quantum mechanics and quantum
computing are the quantum circuits and the   ZX-Calculus
\cite{interacting}. Some variants of the ZX-calculus have been
introduced more recently like the ZW-calculus \cite{zw} and the
ZH-calculus \cite{ZH}. All these languages are defined using
generators (elementary gates) and come with an interpretation functor
which associates with any diagram a pure quantum evolution, i.e. a
morphism in the category of Hilbert spaces. Given a graphical
language, there are generally several ways to represent a quantum
evolution, thus a graphical language is also equipped with an
equational theory which allows to transform a diagram into another
equivalent diagram.  A fundamental property,  generally hard to prove, is the completeness of the language: given two diagrams representing the same quantum evolution, one can be turned into the other using only the transformation rules in the theory.

%
%
%
%

The languages considered have usually been built so as to be able to represent any pure quantum evolution. In this case, the language is called universal for pure quantum mechanics. The hardness of the completeness problem, as well as constraints given by the complexity to physically achieve some gates, focused the research on some restrictions of the languages. On the one hand, finite presentations for the quantum circuits were shown to be complete for some restrictions -- namely Clifford \cite{clifford-circuits}, one-qubit Clifford+T \cite{matsumoto-amano}, two-qubit Clifford+T \cite{SelingerBian}, CNot-dihedral \cite{cnot-dihedral} --, however none of these restrictions is universal, nor approximately universal. Regarding the ZX-calculus, completeness results exist for non-universal restrictions of the ZX-Calculus \cite{pi_2-complete,pi_4-single-qubit,2-qubits-zx,pivoting}, but also for the many-qubit Clifford+T ZX-Calculus \cite{JPV}, which was the first completeness result for an approximately universal fragment of the language. Then complete theories have been introduced  for the universal ZX-Calculus \cite{HNW,JPV-universal,ZXNormalForm,euler-zx} and ZW-Calculus \cite{Amar,HNW}. The completeness of the graphical languages for pure quantum mechanics is one of the main achievements of the categorical approach to quantum mechanics, and is the cornerstone for the application of this formalism in many areas of quantum information processing. 
The ZX-Calculus already proved to be useful for quantum information processing \cite{picturing-qp} (e.g. measure\-ment-based quantum computing  \cite{duncan2013mbqc,mbqc,horsman2011quantum}, quantum codes \cite{chancellor2016coherent,de2017zx,verifying-color-code,duncan2014verifying}, circuit optimisation \cite{duncan2019graph}, foundations \cite{toy-model-graph,duncan2016hopf} ...). Moreover the ZX-calculus can be concretely used through  two softwares: Quantomatic  \cite{kissinger2015quantomatic}  and PyZX \cite{pyzx}.

The existence of complete graphical languages beyond pure quantum mechanics for more general, not necessarily pure, quantum evolutions is an open question that we address in the present paper. 

While pure quantum evolutions correspond to linear maps over Hilbert
spaces, probability distributions over quantum states as well as some
quantum evolutions like discarding a quantum system can be
represented, following the van Neumann approach, by means of density
matrices and completely positive maps. The category of completely
positive maps has been already studied \cite{Selinger:2004:TQP:1014327.1014330}, and in particular the connections between   
the pure and the van Neumann approaches is a central question in
categorical quantum mechanics. Selinger introduced a  construction
called CPM to turn a category for pure quantum mechanics into a category for density matrices and completely positive maps \cite{Selinger-CPM}. Another approach to relate pure quantum mechanics to the general one is the notion of environment structure \cite{mixed-states-axioms,coecke2016pictures,environment}. The notion of \emph{purification} is central in the definition of environment structure. The CPM-construction and the environment structure approaches have been proved to be equivalent  \cite{coecke2016pictures}. 

In terms of graphical languages, the environment structure approach cannot be used in a straightforward way to extend a graphical language beyond pure quantum mechanics. Roughly speaking the environment structure approach provides second order axioms which associates with any equation on arbitrary (non necessarily pure) evolutions an equivalent equation on pure evolutions. Such a second order axiom cannot be easily handled by a equational theory on diagrams. Regarding the CPM-construction, the main property which has been exploited in \cite{picturing-qp} is that CPM($\textbf{C}$) is essentially a subcategory of $\textbf C$, thus one can use a graphical language which has been designed for $\textbf C$ in order to represent morphisms in CPM($\textbf{C}$): Given a complete graphical language for $\textbf C$, we can use a subset of the pure diagrams to represent the evolutions in CPM($\textbf{C}$). The main caveat of this approach is that this subset is not necessarily closed under the equational theory on pure diagrams, and as a consequence does not provide a complete graphical language for CPM($\textbf{C}$).

\vspace{0.1cm}

\noindent {\bf Our contributions.} We introduce a new construction, the \emph{discard construction}, which transforms any $\dagger$-symmetric monoidal category 
	  into a symmetric monoidal category 
	  equipped with a discard map.  Roughly speaking this construction consists in making any isometry causal. Indeed, in quantum mechanics,  the isometries (linear maps $U$ such $U^\dagger \circ U = I$) are known to be causal, i.e. applying $U$ and then discard the subsystem on which it has been applied is equivalent to discarding the subsystem straightaway. Concretely, the discard construction proceeds as follows: first the discard is added to the subcategory of isometries, making the unit of the tensor  a terminal object in this sub-category, as pointed out in \cite{huotuniversal}. Then the discard construction is obtained as the pushout of the resulting category and the initial one.
	  
	  We show that the discard construction does not always produce an environment structure for the original category, and thus is not equivalent to the CPM construction. We show that a necessary and sufficient condition for the two constructions to be equivalent is that the initial category has enough isometries. We show that most of the categories usually used in the context of the categorical quantum mechanics, like $\textbf{FHilb}$ and $\textbf{Stab}$, do have enough isometries, however  $\textbf{Clifford+T}$ does not.

	Finally, we show that the discard construction provide a simple recipe to extend graphical languages beyond pur quantum mechanics. We provide an extension for several graphical languages that we prove to be complete for general quantum operations. 
	

\vspace{0.1cm}
	  
\noindent	  {\bf Structure of the paper.} In section \ref{sec:background}, we review some categorical notions used in categorical quantum mechanics. Section \ref{sec:discard} is dedicated to the definition of the discard construction and the relation with the CPM construction. Finally, in section \ref{sec:graphical} we use  the discard construction to  extend the ZX-calculus to make it complete for general (not necessarily pure) quantum evolutions. The construction is also applied to other graphical languages. 

\section{Background}\label{sec:background}


\subsection{Dagger symmetric monoidal categories}
To avoid any size issue, all our categories are small, the homset of a category $\mathbf{C}$ will be denoted $\mathbf{C}[A,B]$.
	Recall a \emph{strict symmetric monoidal category} (SMC) $\textup{\bf C}$ is a category together with a tensor product bifunctor $\otimes:\mathbf{C}\times \mathbf{C}\to \mathbf{C}$, a unit object $I$ such that $A\otimes I= I\otimes A= A$ and $A\otimes (B\otimes C)=(A\otimes B)\otimes C$, and a symmetry natural isomorphism: $\sigma_{A,B}:A\otimes B\to B\otimes A$ satisfying $\sigma_{A,I}=1_{A}$, $\sigma_{A,B\otimes C}= (1_{B}\otimes \sigma_{A,C}) \circ(\sigma_{A,B}\otimes 1_{C})$, and $\sigma_{A,B}\circ \sigma_{B,A}=1_{B\otimes A}$. A \emph{prop} is an SMC which set of objects is freely spanned by one object. There is an associated notion of strict symmetric monoidal functor $F:\mathbf{C}\to \mathbf{D}$ which preserves unit, tensors and symmetries.

\noindent We will use \emph{string diagram} notations for SMC where morphisms are described as boxes and

\centerline{$g \circ f := \input{./figures/fog.tikz} \qquad
   f \otimes g := \input{./figures/fotimesg.tikz}\qquad
 1_A := \begin{tikzpicture}
	\begin{pgfonlayer}{nodelayer}
		\node [style=none] (0) at (0, 0.2499999) {};
		\node [style=none] (1) at (0, -0.2499999) {};
	\end{pgfonlayer}
	\begin{pgfonlayer}{edgelayer}
		\draw (0.center) to (1.center);
	\end{pgfonlayer}
\end{tikzpicture}\phantom{.}^{A} \qquad
   1_I :=  \input{./figures/empty-diagram.tikz} \qquad
   \sigma_{A,B} := \input{./figures/sigma.tikz}$}


	A $\dagger$-SMC $\mathbf{C}$, is an SMC with an i.o.o. (identity on object) involutive and contravariant SMC-functor $(.)^{\dagger}:\mathbf{C}\to \mathbf{C}$. 
That is, every morphism $f:A\to B$ has a dagger $f^{\dagger}:B\to A$ such that $f^{\dagger\dagger}=f$, moreover the dagger respects the symmetries $\sigma_{A,B}^{\dagger}=\sigma_{B,A}$. The dagger is a central notion in categorical quantum computing and can be used to define specific properties of morphisms:

\vspace{-0.5cm}

\begin{definition}
	$f:A\to B$ is an \emph{isometry} if $f^{\dagger}\circ f=1_{A}$, i.e.~~$\input{./figures/iso.tikz}\eq{}\begin{tikzpicture}
	\begin{pgfonlayer}{nodelayer}
		\node [style=none] (0) at (0, 0.2499999) {};
		\node [style=none] (1) at (0, -0.2499999) {};
	\end{pgfonlayer}
	\begin{pgfonlayer}{edgelayer}
		\draw (0.center) to (1.center);
	\end{pgfonlayer}
\end{tikzpicture}~$.
\end{definition}

\vspace{-0.4cm}

In this paper most of the categories considered are furthermore compact closed: 
	A dagger compact category ($\dagger$-CC) is a $\dagger$-SMC where every object $A$ has a dual object $A^{*}$ such that for all objects $A$, there are two morphisms $\input{./figures/cup-1.tikz}:A\otimes A^*\to I$ and $\input{./figures/cap-1.tikz}:I\to A^*\otimes A$ satisfying  
	$\input{./figures/BMC-snake-1.tikz}$, $\input{./figures/BMC-snake-2.tikz}$ and $\left(\input{./figures/cup-1.tikz}\right)^{\dagger}=\input{./figures/SCC-swap-cap.tikz}$. 


\subsection{Examples}
We are considering two kinds of SMCs in this paper: the categories of quantum evolutions and the graphical languages. 

\noindent {\bf Quantum evolutions.} Pure quantum evolutions correspond the category of Hilbert spaces. We will consider various subcategories of it: 
$\mathbf{FHilb}$ is the  category of finite dimensional Hilbert spaces which objects are $\mathbb C^n$ and morphisms are linear maps. Its tensor is the usual tensor product of vector spaces and its dagger is the adjoint with respect to the usual scalar product. It is the mathematical model for pure quantum mechanics. In quantum information processing, the quantum data are usually carried by qubits, hence $\mathbf{Qubit}$ is the full subcategory of $\mathbf{FHilb}$ with objects of the form $\mathbb{C}^{2^{n}}$. $\textup{\bf Stab}$ is the sub-category of $\mathbf{Qubit}$ which is finitely generated by the Clifford operators: H, S, CNot, the state $\ket 0$, the projector $\bra 0$, and the scalars $2$ and $i$ where:

\centerline{ $\textup H= \frac{1}{\sqrt{2}} {\footnotesize  \begin{pmatrix}1 & 1\\1 & -1\end{pmatrix}}\quad \textup S= {\footnotesize  \begin{pmatrix}1 & 0\\0 & i\end{pmatrix} }\quad \textup{CNot}= {\footnotesize  \begin{pmatrix}1 & 0&0&0\\0&1 & 0&0\\0&0&0&1\\0&0&1&0\end{pmatrix}}\quad \ket 0 =  {\footnotesize  \begin{pmatrix}1 \\0 \end{pmatrix}} \quad \bra 0 =  {\footnotesize  \begin{pmatrix}1 & 0\end{pmatrix}}$}

Those are amongst the most commonly used gates in quantum computation see \cite{nielsen_chuang_2010} for details. 
$\mathbf{Clifford{+}T}$ is the same as $\mathbf{Stab}$ but with the additional generator $T = {1~~~0 \choose ~0~e^{i\frac\pi4}}$, the morphisms of $\mathbf{Clifford{+}T}$ are exactly the matrices with entries in the ring $\mathbb Z[i, \frac 1{\sqrt 2}]$ \cite{JPV}. Contrary to $\textup{\bf Stab}$, $\mathbf{Clifford{+}T}$ is \emph{approximately universal} in the sense that $\forall n,m{\in} \mathbb N$, $\forall f{\in} \mathbf{Qubit}[\mathbb C^{2^n}, \mathbb C^{2^m}]$ and  $\forall \epsilon {>}0$, there exists $g\in \mathbf{Clifford{+}T}[\mathbb C^{2^n}, \mathbb C^{2^m}]$ such that $||f-g||{<}\epsilon$. $\mathbf{FHilb}$,  $\mathbf{Qubit}$, $\mathbf{Clifford{+}T}$, and $\textup{\bf Stab}$ are all $\dagger$-CC. Notice that $\mathbf{Qubit}$, $\mathbf{Clifford{+}T}$, and $\textup{\bf Stab}$ are props, but $\mathbf{FHilb}$ is not. 

	
	Probability distributions over pure quantum states as well as some quantum evolutions like discarding a quantum system are not \emph{pure} but can be represented, following the van~Neumann approach, by means of density matrices and completely positive maps. Let $\mathbf{CPM}$ be the category of finite dimension completely positive maps which objects are $\mathbb C^n$ and $\mathbf{CPM}[\mathbb C^n,\mathbb C^m] =\{U:\mathbb C^{n\times n} \to \mathbb C^{m\times m} ~|~ U$ is a completely positive linear map$\}$. Similarly to the pure case, one can define various subcategories of $\mathbf{CPM}$. Notice it can be achieved by the CPM construction described in the next section. 


%

\vspace{0.2cm}

\noindent {\bf Graphical languages.} The second kind of categories we are considering in this paper are graphical languages. They are props which come with interpretation functors defining their semantics. A prop is in fact the equivalent of Lawvere theories for symmetric monoidal theories. They can be presented by generators and relations as one would do for usual theories, see \cite{PhD.Zanasi} and \cite{baez-PROP} for a detailed discussion.

%

\begin{definition}
	A \emph{graphical language} $\mathcal{G}$ is a prop presented by a set of \emph{generators} $\Sigma$ and a set of \emph{equations} $E$ together with a function $\interp{.}:\Sigma \to \hom(S)$ called the \emph{interpretation} of $\mathcal{G}$ in $S$. 
	$\mathcal{G}$ is said to be \emph{sound} if $\interp{.}$ defines an interpretation functor $\interp{.}:\mathcal{G} \to S$, and \emph{universal} (resp. \emph{complete}) when this functor is surjective (resp. faithful).
\end{definition}

The ZX-, ZW- and ZH-calculi or the quantum circuits are examples of such categories with semantics in $\mathbf{Qubit}$.

\subsection{Environment structures and CPM-construction}

Connecting the Hilbert approach -- for pure quantum mechanics -- and the van Neumann approach -- for open systems --  is a central question in categorical quantum mechanics. Selinger pointed out that any $\dagger$-CC for pure quantum mechanics can be turned into a category for density matrices and completely  positive maps via the CPM construction \cite{Selinger-CPM}:

\begin{definition}
	Given a $\dagger$-CC $\mathbf{C}$, let $\operatorname{CPM}(\mathbf{C})$ be the $\dagger$-CC with the same objects as $\mathbf{C}$ such that $\operatorname{CPM}(\mathbf{C})[A,B]=\left\lbrace\input{./figures/BCC-cpm-tensor.tikz},f\in~ \textup{\bf C} [A,B\otimes C] \right\rbrace$, where $\input{./figures/BCC-adjoint.tikz}$.
\end{definition}

Applying it to $\mathbf{FHilb}$ one obtains the category $\mathbf{CPM}$ of completely positives maps. The CPM construction can also be applied to $\mathbf{Qubit}$, $\mathbf{Clifford{+}T}$, and $\textup{\bf Stab}$. Notice that the CPM-construction has been then extended to the non necessarily compact categories \cite{}. 

\vspace{0.2cm}
Another approach to relate pure quantum mechanics to the general one is the notion of environment structure \cite{mixed-states-axioms,coecke2016pictures,environment}. The notion of \emph{purification} is central in the definition of environment structure. 
Intuitively, it means that (1) there is a discard morphism for every object; (2) any morphism can be purified, i.e. decomposed into a pure morphism followed by a discarding map, and (3) this purification is essentially unique. More formally:

\begin{definition}
	An environment structure for a $\dagger$-CC $\mathbf{C}$ is an CC $\overline{\mathbf{C}}$ with the same objects as $\mathbf{C}$, an i.o.o SMC-functor $\iota:\mathbf{C}\to\overline{\mathbf{C}}$ and for each object $A$ a morphism $\ground_{A}:A\to I$ such that:
	\begin{itemize}
		\item[$(1)$] $\disc{I}=1_{I}$, and for all $A,B:\mathbf{C}$, $\disc{A}\otimes \disc{B}=\disc{A\otimes B}$.
		\item[$(2)$] For all $f:A\to B$ in $\overline{\mathbf{C}}$, there is an $f':A\to B\otimes X$ in $\mathbf{C}$ such that:
		$\input{./figures/excond0.tikz}=\input{./figures/excond1.tikz}$
		\item[$(3)$] For any $f:A\to B\otimes X$ and $g:A\to B\otimes Y$ in $\mathbf{C}$: $f\cp  g \Leftrightarrow \input{./figures/uncond0.tikz}=\input{./figures/uncond1.tikz}$\\
		
	\end{itemize}
	where the relation $\cp$ is defined as:
		$f\cp  g \Leftrightarrow \input{./figures/simcpm0.tikz}=\input{./figures/simcpm1.tikz}$
\end{definition}
Notice that $\cp$ is technically not a relation on morphisms but on
tuples $(A,B,X,f)$ with $f\in \mathbf{C}[A,B\otimes X]$:
$(A,B,X,f) \cp (C,D,Y,g)$ if $A=C, B=D$ and $f,g$ satisfy the
graphical condition represented above. As an abuse of notation, we
write $f \cp g$, as the  other components of the
tuple will be usually obvious from context. We will do the same for
our relation $\iso$ below.

$\textup{\bf CPM}$ is actually an environment structure for the category $\textup{\bf FHilb}$, and more generally for any $\dagger$-CC $\mathbf{C}$, $\operatorname{CPM}(\mathbf{C})$ is an environment structure for $\mathbf{C}$ 	and conversely any environment structure for $\mathbf{C}$ is equivalent to $\operatorname{CPM}(\mathbf{C})$ \cite{coecke2016pictures}. Actually one can notice that $\operatorname{CPM}(\mathbf{C})[A,B]$ is nothing but the set of equivalent classes of $\cp$.




The notion of environment structures has also be generalisation to the non compact case \cite{coecke2016pictures}.  We chose here to focus on the compact case.


\section{The Discard Construction}\label{sec:discard}
%
%
We introduce a new construction, the discard construction which consists in adding a discard map for every object of a $\dagger$-SMC, and thus intuitively transforming a category for pure quantum mechanics into a category for general quantum evolutions. 

Causality is a central notion in quantum mechanics which has been axiomatised using a discard map as follows \cite{kissinger2017categorical}: $f:A\to B$ is \emph{causal} if and only if $\begin{tikzpicture}[scale=0.5]
	\begin{pgfonlayer}{nodelayer}
		\node [style=none] (7) at (0, 0.75) {};
		\node [style=none] (10) at (0, -0.75) {$\ground$};
		\node [style=Box] (11) at (0, 0) {$f$};
	\end{pgfonlayer}
	\begin{pgfonlayer}{edgelayer}
		\draw (10.center) to (11.center);
		\draw (7.center) to (11.center);
	\end{pgfonlayer}
\end{tikzpicture}
 =\begin{tikzpicture}[scale=0.5]
	\begin{pgfonlayer}{nodelayer}
		\node [style=none] (10) at (0, -0.25) {$\ground$};
		\node [style=none] (11) at (0, 0.25) {};
	\end{pgfonlayer}
	\begin{pgfonlayer}{edgelayer}
		\draw (11.center) to (10.center);
	\end{pgfonlayer}
\end{tikzpicture}
$.  Among the pure quantum evolutions, the isometries are causal evolutions. The discard construction essentially consists in making any isometry causal.  Thus, whereas  the CPM construction relies on completely positive maps and the environment structures on the concept of purification, 
the discard construction relies on causality. 

\subsection{Definition}
We introduce the new construction in three steps. First, given a  $\dagger$-SMC, one can consider its subcategory of isometries:

\begin{definition}
	Given a $\dagger$-SMC $\mathbf{C}$, $\ciso{\mathbf{C}}$ is the subcategory with the same object as $\mathbf{C}$ and isometries as morphisms, i.e. for all $A,B:\mathbf{C}$, $\ciso{\mathbf{C}}[A,B]=\{f:\mathbf{C}[A,B], f^{\dagger}\circ f=1_{A}\}$.
\end{definition}

Notice that $\ciso{\mathbf{C}}$ is a SMC but usually not a $\dagger$-SMC. Any $\dagger$-SMC-functor $F:\mathbf{C}\to \mathbf{D}$ between two $\dagger$-SMC can be restricted to their subcategories of isometries leading to an SMC-functor $\ciso{F}:\ciso{\mathbf{C}}\to \ciso{\mathbf{D}}$. Thus there is a restriction functor $\operatorname{iso}:\dagger\mathbf{-SMC}\to \mathbf{SMC}$. Remark that this functor preserves fullness and faithfulness.
One always has an inclusion i.o.o. faithful SMC-functor: $\ciso{i}:\mathbf{C}\to \ciso{\mathbf{C}}$.

In quantum mechanics, isometries are causal evolutions, i.e. applying an isometry and then discarding all outputs is equivalent to discarding the inputs straight away. As pointed out in \cite{huotuniversal}, adding discard maps to the category of isometries would make $I$ a terminal object. Such a category is said to be \emph{affine symmetric monoidal category} (ASMC). We define the affine completion of an SMC:

\begin{definition}
	Given an SMC $\mathbf{C}$, we define $\mathbf{C^{!}}$ as $\mathbf{C}$ with an additional morphism $\term{A}:A\to I$ for each object $A:\mathbf{C}$, such that, 
for all $f:\ciso{\mathbf{C}}[A,B]$, $\term{B}\circ f  = \term{A}$. This makes $I$ a terminal object in $\cterm{\mathbf{C}}$, and then $\cterm{\mathbf{C}}$ is an ASMC.
\end{definition}

Again given a functor $F:\mathbf{C}\to \mathbf{D}$, one can define a functor $\cterm{F}:\mathbf{C}^{!}\to \mathbf{D}^{!}$ by $F^{!}(\term{A})=\term{\cterm{i}(F(A))}$ and $F^{!}(f)=i^{!}(F(f))$ for the other morphisms.
In \cite{huotuniversal}, Huot and Staton show that $\mathbf{CPTPM}$, the category of completely positive trace preserving maps, is equivalent to $\cterm{\ciso{\mathbf{FHilb}}}$, thus giving a caracterisation of it via a universal property. We extend this idea to non-trace preserving maps by proceeding to a local affine completion of the subcategory of isometries.

We define the category $\cdisc{\mathbf{C}}$ as the pushout of $\mathbf{C}$ and $\cterm{\ciso{\mathbf{C}}}$:
\begin{definition}
	Given a $\dagger$-SMC $\mathbf{C}$, $\cdisc{\mathbf{C}}$ is defined as the pushout:
	\begin{center}
		\begin{tikzpicture}
		\node (Ciso) at (0,0) {$\ciso{\mathbf{C}}$};
		\node (Cisog) at (0,-1.5) {$\cterm{\ciso{\mathbf{C}}}$};
		\node (C) at (2,0) {$\mathbf{C}$};
		\node (Cg) at (2,-1.5) {$\cdisc{\mathbf{C}}$};
		\draw[right hook->] (Ciso) -- (C) node[midway,above] {$\ciso{i}$};
		\draw[->] (Ciso) -- (Cisog) node[midway,left] {$\cterm{i}$};
		\draw[right hook->] (Cisog) -- (Cg) node[midway,below] {$\iota_{\cterm{\ciso{\mathbf{C}}}}$};
		\draw[->] (C) -- (Cg) node[midway,right] {$\iota_{\mathbf{C}}$};
		\node[above left=-0.6em and -0.6em of Cg] {\rightangle};
		\end{tikzpicture}
	\end{center}
	
\end{definition}

Classical results on enriched categories show that the pushout of two
SMCs always exists. As all our functors are i.o.o, we can also
describe it simply combinatorially.
The objects of $\cdisc{\mathbf{C}}$ are the same as $\mathbf{C}$. Its morphisms are equivalence classes generated by formal composition and tensoring of morphisms in $\cterm{\ciso{\mathbf{C}}}$ and $\mathbf{C}$. The equivalence relation is generated by the equations of both categories augmented with equations $\cterm{i}(f)=\ciso{i}(f)$ for all $f$ in $\ciso{\mathbf{C}}$. The functors $\iota_{\mathbf{C}}$ and $\iota_{\cterm{\ciso{\mathbf{C}}}}$ are the natural ways to embed $\mathbf{C}$ and $\cterm{\ciso{\mathbf{C}}}$.
We will see those formal compositions as string diagrams whose components are morphisms of $\mathbf{C}$ and $\cterm{\ciso{\mathbf{C}}}$ wired to each others. Two diagrams represent the same morphism if we can rewrite one into the other applying the equations of both categories and $\cterm{i}(f)=\ciso{i}(f)$ for all $f$ in $\ciso{\mathbf{C}}$. This forms a well defined SMC.

Since the only morphisms in $\ciso{\mathbf{C}}$ which are not
identified with the morphisms of $\mathbf{C}$ are those that contain $\term{A}$, we can see $\cdisc{\mathbf{C}}$ as $\mathbf{C}$ augmented with discard maps which delete isometries.

\begin{definition}
	The discard map on an object $A$ is defined in $\cdisc{\mathbf{C}}$ by $\disc{A}\coloneqq \iota_{\cterm{\ciso{\mathbf{C}}}}(\term{A})$.
\end{definition}

Notice, that for any isometry $f:A\to B$ in $\cdisc{\mathbf{C}}$, $\begin{tikzpicture}[scale=0.5]
	\begin{pgfonlayer}{nodelayer}
		\node [style=none] (7) at (0, 0.75) {};
		\node [style=none] (10) at (0, -0.75) {$\ground$};
		\node [style=Box] (11) at (0, 0) {$f$};
	\end{pgfonlayer}
	\begin{pgfonlayer}{edgelayer}
		\draw (10.center) to (11.center);
		\draw (7.center) to (11.center);
	\end{pgfonlayer}
\end{tikzpicture}
 =\begin{tikzpicture}[scale=0.5]
	\begin{pgfonlayer}{nodelayer}
		\node [style=none] (10) at (0, -0.25) {$\ground$};
		\node [style=none] (11) at (0, 0.25) {};
	\end{pgfonlayer}
	\begin{pgfonlayer}{edgelayer}
		\draw (11.center) to (10.center);
	\end{pgfonlayer}
\end{tikzpicture}
$, thus any isometry is causal. 


\subsection{Relation to environment structures and $\operatorname{CPM}$}

In order to compare the $\mathbf{C}^{\sground}$ construction with environment structures and the $\operatorname{CPM}$ construction we need to study in details the purification process in $\mathbf{C}^{\sground}$. 
First notice that any morphism of $\mathbf{C}^{\sground}$ admits a purification:
\begin{lemma}\label{lem:pur} Let $\mathbf{C}$ be a $\dagger$-SMC , 
	For all $f:\mathbf{C}^{\sground}[A,B]$, there is an $X:\mathbf{C}$ and an $f':\mathbf{C}[A,B\otimes X]$ such that $\input{./figures/excond0.tikz}= \input{./figures/excond2.tikz}$. 
\end{lemma}
\noindent The proof is in appendix at page \pageref{prf:pur}.
\vspace{0.15cm}

The purification needs not be unique, however it satisfies an essential uniqueness condition. To state it we define the relation $\iso $.:

\begin{definition}
	Let $\mathbf{C}$ be a $\dagger$-SMC, and two morphisms $f:A\to B\otimes X$, $g:A\to B\otimes Y$, $f\iso g$ if there are two isometries $u:X\to Z$ and $v:Y\to Z$, such that $\input{./figures/simis0-s.tikz}=\input{./figures/simis1-s.tikz}$.
\end{definition}

Notice that the relation $\iso $ is not transitive, thus we consider $\isop$  its transitive closure to make it an equivalence relation. It is easy to show that if $f\isop g$ then $f$ and $g$ purify the same morphism of $\mathbf{C}^{\sground}$. The converse is also true:

\begin{lemma}\label{lem:unicity}
	For all $f:A\to B\otimes X$ and $g:A\to B\otimes Y$: $f\isop g \Leftrightarrow \input{./figures/lm0-s.tikz}=\input{./figures/lm1-s.tikz}$
\end{lemma}

\noindent The proof is in appendix at page \pageref{prf:unicity}.
\vspace{0.15cm}

So the purification is unique up to $\isop$. Lemma \ref{lem:unicity} also gives an alternative definition of $\cdisc{\mathbf{C}}$ which relates more easily to the $\operatorname{CPM}$ construction. It is the same construction as $\operatorname{CPM}$ with $\cp$ replaced by $\isop$. In other words $\cdisc{\mathbf{C}}[A,B]$ is the set of equivalent classes of $\isop$.

As we have introduced a new discard construction, a natural question is whether $\mathbf{C}^{\sground}$ is an environment structure for $\mathbf{C}$. 
To be an environment structure, three conditions are required. The
first two are satisfied: $\mathbf{C}^{\sground}$ has a discard
morphism for every object, and every morphism can be purified. The
third one is the uniqueness of the purification: according to the
definition of the environment structures, $f$ and $g$ purify the same
morphism if and only if $f\cp  g$ whereas according to Lemma
\ref{lem:unicity}, $f$ and $g$ purify the same morphism if and only if
$f\isop g$. As a consequence $\mathbf{C}^{\sground}$ is an environment
structure for $\mathbf{C}$ if and only if $\cp   = \isop$.
It turns out that one of the inclusions is always true:
\begin{lemma} \label{lem:inclu}
For any $\dagger$-SMC category $\mathbf{C}$, we have $\isop \subseteq \cp$.
\end{lemma}
\noindent The proof is in appendix at page \pageref{prf:inclu}.

As a consequence, if $\cp \not= \isop$, it means that there are some
morphisms $f,g$ that are equal in $\cp$ but cannot be proved equal in
$\isop$. Intuitively it means the category has not enough isometries
to prove those terms equal, which leads to the following definition:

\begin{definition}
  A $\dagger$-SMC category $\mathbf{C}$ has \emph{enough isometries} if the
  equivalences relations $\cp$ and $\isop$  of $\mathbf{C}$ are equal.
\end{definition}

\begin{lemma}\label{thm:cpenv}
 Given a $\dagger$-SMC $\mathbf{C}$, the following properties are equivalent:
 \begin{enumerate}
 \item $\mathbf{C}$ has enough isometries;
 \item $\mathbf{C}^{\sground}$ is an environment structure for  $\mathbf{C}$;
 \item $\mathbf{C}^{\sground}\simeq \operatorname{CPM}(\mathbf{C})$.
 \end{enumerate}
 \end{lemma}

\noindent The proof is in appendix at page \pageref{prf:cpenv}.
\vspace{0.15cm}

%
%
%
%
%
%
%

Notice that if $\textbf C$ has enough isometries, the discard construction provides a definition of $\operatorname{CPM}(\textbf C)$ via a universal property. This gives a more direct way to built the environment, avoiding to deal with the equivalence classes of the $\operatorname{CPM}$ construction. 

\begin{remark}
 Let's focus for a moment on the category $\operatorname{Causal}\operatorname{CPM}(\mathbf{C})$ of causal maps, that is the subcategory of maps cancelled by the discards in $\operatorname{CPM}(\mathbf{C})$. We have that:
 $\cp  \subseteq\isop \Rightarrow \cterm{\ciso{\mathbf{C}}} \simeq \operatorname{Causal}\operatorname{CPM}(\mathbf{C})$.
	In fact by Lemma \ref{thm:cpenv}, $\operatorname{CPM}(\mathbf{C})\simeq \mathbf{C}^{\sground}$, and then the subcategory $\operatorname{Causal}\operatorname{CPM}(\mathbf{C})$ is equivalent to the subcategory of maps cancelled by the discards in $\mathbf{C}^{\sground}$ which is equivalent to $\cterm{\ciso{\mathbf{C}}}$.
	$\operatorname{Causal}\operatorname{CPM}(\mathbf{FHilb})$ being exactly $\mathbf{CPTPM}$, we have recovered the result of \cite{huotuniversal}.
\end{remark}

\subsection{Examples}

We consider the usual subcategories of $\mathbf{FHilb}$ used for pure quantum mechnanics and show in each case whether the discard construction produces an environment structure or not. 
First of all, thanks to the Stinespring dilation theorem, $\mathbf{FHilb}^{\sground}$ is not only an environment structure for $\mathbf{FHilb}$, but the relation $\iso$ is also transitive in this case:



\begin{proposition}\label{lem:CPFH}
$\mathbf{FHilb}^{\sground}$ is an environment structure for
$\mathbf{FHilb}$. Furthermore $\isop = \iso$.
\end{proposition}

\noindent The proof is in appendix at page \pageref{prf:CPFH}.
\vspace{0.15cm}

When dealing with graphical languages we will be more interested in the full subcategory $\mathbf{Qubit}$ of $\mathbf{FHilb}$:

\begin{proposition}\label{lem:CPQ}
	$\mathbf{Qubit}^{\sground}$ is an environment structure for
$\mathbf{Qbit}$.
\end{proposition}

\noindent The proof is in appendix at page \pageref{prf:CPQ}.
\vspace{0.15cm}
 
Notice that in general, the property of having enough isometries does not transfer to full subcategories: 
 If $\mathbf{D}$ is a full subcategory of
$\mathbf{C}$, we might have $f \isop g$ on $\mathbf{C}$ but
$f \not\isop g$ on $\mathbf{D}$. This could happen for two reasons:
First the chain of intermediate morphisms that prove that $f \isop g$
might live outside of $\mathbf{D}$. Second, the isometries that ``prove''
that $f \isop g$ on $\mathbf{C}$ might have codomain outside of $\mathbf{D}$.

If our category is not a full subcategory, then all hell breaks loose, and finding
conditions that guarantees that $\mathbf{C}^{\sground}$  is an environment
structure for $\mathbf{C}$ is not easy.

For subcategories of $\mathbf{Qubit}$, necessary conditions can be given. This
category has the peculiarity that $\cdot^*$ is the identity on object and
that $f^{**} = f$ for all morphisms ($\cdot^*$ maps a matrix to its
conjugate matrix). In particular, for any
state $\phi: I \rightarrow I \otimes X$, we have $\phi^* \cp
\phi$. Indeed
\input{./figures/phi-star-phi.tikz}.

So a necessary condition for a subcategory of $\mathbf{Qubit}$ to
behave nicely is that for all states $\phi$, we have $\phi^* \isop \phi$. 
This is the case in $\mathbf{Stab}$: Given a stabilizer state $\phi$, there always exists a
stabilizable unitary $U$ s.t. $U\phi = \phi^*$. In fact:

\begin{proposition}\label{lem:CPS}
	${\mathbf{Stab}}^{\sground}$ is an environment structure for $\mathbf{Stab}$.
\end{proposition}

\noindent The proof is in appendix at page \pageref{prf:CPS}.
\vspace{0.15cm}

      The main idea of the proof is to use the map/state duality, and structural results about
      bipartite stabilizer states \cite{clifford-state-partition}.

No such unitary exist in general in $\mathbf{Clifford{+}T}$: For
almost all states $\phi$, there are no unitary $U$ (and even no
morphism at all) s.t. $U \phi = \phi^*$.  $\mathbf{Clifford{+}T}$
therefore has not enough isometries:

\begin{proposition}\label{lem:CPC}
	$({\mathbf{Clifford{+}T}})^{\sground}$ is not an environment structure
        for $\mathbf{Clifford{+}T}$. More precisely, there exists a state $\phi$
        s.t. $\phi \cp \phi^*$ but $\phi \not\isop \phi^*$. One can
        take for example $\phi = 1 + 2i$ (in this case $\phi$ is a state
        with no input and outputs, hence a scalar).
\end{proposition}

\noindent The proof is in appendix at page \pageref{prf:CPC}.


\section{Application to the ZX-Calculus and other graphical languages} \label{sec:graphical}

We now focus on the behavior of interpretation functors with respect to the discard construction. The discard construction defines a functor $(\_)^{\sground}:\dagger{-}\mathbf{SMC}\to \mathbf{SMC}$. Indeed, given a $\dagger$-SMC functor $F$, $\ciso{F}$ and $\cterm{\ciso{F}}$ uniquely define a functor $\cdisc{F}$ by pushout.
\begin{center}
	\begin{tikzpicture}[scale=1.4]
	\node (Diso) at (1,0.75) {$\ciso{\mathbf{D}}$};
	\node (Disog) at (1,-0.75) {$\cterm{\ciso{\mathbf{D}}}$};
	\node (D) at (3,0.75) {$\mathbf{D}$};
	\node (Dg) at (3,-0.75) {$\mathbf{D^{\sground}}$};
	\path[right hook->] (Diso) edge (D);
	\path[->] (Diso) edge (Disog);
	\path[right hook->] (Disog) edge (Dg);
	\path[->] (D) edge (Dg);
	\node[above left=-0.6em and -0.6em of Dg] {\rightangle};
	
	\node[circle,draw=white, fill=white, inner sep=0pt,minimum size=2pt] at (1,0) {};
	\node[circle,draw=white, fill=white, inner sep=0pt,minimum size=2pt] at (2,-0.75) {};
	
	\node (Ciso) at (0,0) {$\ciso{\mathbf{C}}$};
	\node (Cisog) at (0,-1.5) {$\cterm{\ciso{\mathbf{C}}}$};
	\node (C) at (2,0) {$\mathbf{C}$};
	\node (Cg) at (2,-1.5) {$\cdisc{\mathbf{C}}$};
	\path[right hook->] (Ciso) edge (C);
	\path[->] (Ciso) edge (Cisog);
	\path[right hook->] (Cisog) edge (Cg);
	\path[->] (C) edge (Cg);
	\node[above left=-0.6em and -0.6em of Cg] {\rightangle};
	
	\draw[->] (Ciso) -- (Diso) node[midway, sloped, above] {$F_{\textnormal{iso}}$};
	\draw[dashed, ->] (Cisog) edge node[midway,sloped,above] {$F_{\textnormal{iso}}^{!}$} (Disog) ;
	\draw[->] (C) -- (D) node[midway, above,sloped] {$F$};
	\draw[dashed, ->] (Cg) edge node[below,sloped]  {$\cdisc{F}$}  (Dg);
	\end{tikzpicture}
\end{center}
The following lemma and theorem are the main tools to apply the discard construction to graphical languages:

\begin{lemma}\label{lem:Fis}
	If $F$ is faithful and if $\ciso{F} : \ciso{\mathbf{C}}\to
        \ciso{\mathbf{D}}$ is surjective, then $F(f)\isop F(g) \Rightarrow f\isop g$.
\end{lemma}
\noindent The proof is in appendix at page \pageref{prf:Fis}.

\begin{theorem}
	\label{thm:ground}
	Let $\mathbf{C}$ and $\mathbf{D}$ be two $\dagger$-SMCs and
        $F:\mathbf{C}\to\mathbf{D}$ a $\dagger$-SMC-functor. If $F$ is
        faithful and if $\ciso{F} : \ciso{\mathbf{C}} \to
        \ciso{\mathbf{D}}$ is surjective, then $\cdisc{F}:\cdisc{\mathbf{C}} \to \cdisc{\mathbf{D}}$ is faithful. If furthermore $F$ is surjective then $\cdisc{F}$ is surjective and faithful.
\end{theorem}
\noindent The proof is in appendix at page \pageref{prf:ground}.
\vspace{0.15cm}

Notice that the hypothesis on $\ciso{F}$ is very strong, as it makes
it an isomorphism: We want it to be surjective as we do not
want to lose even one isometry.  In particular we do not know if the
theorem still applies if $F$ is merely an equivalence of category.


Reformulating for graphical languages this gives:

\begin{corollary}[of Theorem \ref{thm:ground}] \label{cor:ground} Given a $\dagger$-CC $\mathbf{C}$ with enough isometries, 
	if $\mathcal{G}$ is a $\dagger$-CC universal complete graphical language for $\mathbf{C}$ then $\cdisc{\mathcal{G}}$ is a universal complete language for $\operatorname{CPM(\mathbf{C})}$.
\end{corollary}

This provides a general recipe. We start by a universal complete graphical language $\mathcal{G}$. We build $\cdisc{\mathcal{G}}$, by Theorem \ref{thm:ground}, $\cdisc{\interp{.}}:\cdisc{\mathcal{G}}\to \cdisc{\mathbf{C}}$ is full and faithful. Furthermore $\cdisc{\mathbf{C}}\simeq \operatorname{CPM}(\mathbf{C})$. $\cdisc{\mathcal{G}}$ as a prop can be presented by adding one new generator $\ground$ to the signature $\Sigma$ and one equation for each isometry of $\mathcal{G}$. In general, if one is provided with a spanning set of the isometries, the number of equations can be drastically reduced. We just need one equation for each element of this set. We then obtain a universal complete graphical language.  

We will now briefly review the ZX-calculus and some of its twin languages. They are all universal and complete for subcategories of $\mathbf{Qubit}$. Each time we will apply the recipe with a well chosen spanning set and provide the additional axioms involving $\ground$. We will not discuss minimality, i.e.~if adding these new axioms can help to simplify others.

\subsection{The ZX-calculus}
The ZX-Calculus was introduced in \cite{interacting} by Coecke and Duncan for pure quantum evolutions. 
It is a $\dagger$-compact prop generated by:
\[R_Z^{(n,m)}(\alpha):n\to m::\input{./figures/gn-alpha.tikz}\qquad\qquad R_X^{(n,m)}(\alpha):n\to m::\input{./figures/rn-alpha.tikz}\qquad\qquad H:1\to 1::~~\begin{tikzpicture}
	\begin{pgfonlayer}{nodelayer}
		\node [style=H box] (0) at (0, 0) {};
		\node [style=none] (1) at (0, 0.275) {};
		\node [style=none] (2) at (0, -0.275) {};
	\end{pgfonlayer}
	\begin{pgfonlayer}{edgelayer}
		\draw (2.center) to (1.center);
	\end{pgfonlayer}
\end{tikzpicture}
\]
and the two compositions: spacial ($.\otimes.$) and sequential ($.\circ.$).
The symmetric and compact structure are provided by $\sigma:2\to 2::\input{./figures/crossing.tikz}$, $\epsilon:2\to 0::\begin{tikzpicture}
	\begin{pgfonlayer}{nodelayer}
		\node [style=none] (0) at (-0.2500001, 0.2500001) {};
		\node [style=none] (1) at (0.2500001, 0.2500001) {};
	\end{pgfonlayer}
	\begin{pgfonlayer}{edgelayer}
		\draw [bend right=90, looseness=1.75] (0.center) to (1.center);
	\end{pgfonlayer}
\end{tikzpicture}$ and $\eta:0\to 2::\begin{tikzpicture}
	\begin{pgfonlayer}{nodelayer}
		\node [style=none] (a0) at (-0.2500001, -0) {};
		\node [style=none] (a1) at (0.2500001, -0) {};
		\node [style=none] (a2) at (0, 0.25) {};
	\end{pgfonlayer}
	\begin{pgfonlayer}{edgelayer}
		\draw [bend left=90, looseness=1.75] (a0.center) to (a1.center);
	\end{pgfonlayer}
\end{tikzpicture}$.


To simplify, the red and green nodes will be represented empty when holding a 0 angle:
\[ \scalebox{0.9}{\input{./figures/gn-empty-is-gn-zero.tikz}} \qquad\text{and}\qquad \scalebox{0.9}{\input{./figures/rn-empty-is-rn-zero.tikz}} \]

The language is universal \cite{interacting}. So far, it has two complete axiomatisations \cite{HNW,JPV-universal}. One is given in Appendix in Figure \ref{fig:ZX_rules}, but any complete axiomatisation will suffice. Some of the main axioms are:
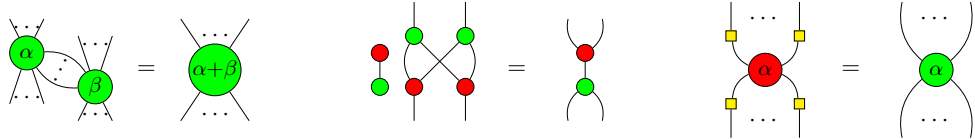
\[\scalebox{0.9}{\input{./figures/spider-1.tikz}\qquad\qquad\input{./figures/b2s.tikz}\qquad\qquad\input{./figures/h2.tikz}}\]


ZX-diagrams represent quantum evolutions, so there exists a functor $\interp{.}:\text{ZX}\to \mathbf{Qubit}$, called the \emph{standard interpretation}, which 
associates to any diagram $D:n\to m$ a linear map $\interp{D}:\mathbb{C}^{2^n}\to\mathbb{C}^{2^m}$ inductively defined as follows:\\
{\setlength{\arraycolsep}{2pt}\def\arraystretch{0.7}
\begin{minipage}{\columnwidth}
\titlerule{$\interp{.}$}
$$\interp{D_1\otimes D_2}:=\interp{D_1}\otimes\interp{D_2}\qquad\qquad\interp{D_2\circ D_1}:=\interp{D_2}\circ\interp{D_1}$$
\end{minipage}
$$\interp{\input{./figures/empty-diagram.tikz}~}:=\begin{pmatrix}
1
\end{pmatrix} \qquad\qquad
\interp{~\begin{tikzpicture}
	\begin{pgfonlayer}{nodelayer}
		\node [style=none] (0) at (0, 0.2499999) {};
		\node [style=none] (1) at (0, -0.2499999) {};
	\end{pgfonlayer}
	\begin{pgfonlayer}{edgelayer}
		\draw (0.center) to (1.center);
	\end{pgfonlayer}
\end{tikzpicture}~~}:= \begin{pmatrix}
1 & 0 \\ 0 & 1\end{pmatrix}\qquad\qquad
\interp{~\begin{tikzpicture}
	\begin{pgfonlayer}{nodelayer}
		\node [style=H box] (0) at (0, 0) {};
		\node [style=none] (1) at (0, 0.275) {};
		\node [style=none] (2) at (0, -0.275) {};
	\end{pgfonlayer}
	\begin{pgfonlayer}{edgelayer}
		\draw (2.center) to (1.center);
	\end{pgfonlayer}
\end{tikzpicture}
~}:= \frac{1}{\sqrt{2}}\begin{pmatrix}1 & 1\\1 & -1\end{pmatrix}$$
$$\interp{\raisebox{-0.25em}{$\begin{tikzpicture}
	\begin{pgfonlayer}{nodelayer}
		\node [style=none] (0) at (-0.2500001, 0.2500001) {};
		\node [style=none] (1) at (0.2500001, 0.2500001) {};
	\end{pgfonlayer}
	\begin{pgfonlayer}{edgelayer}
		\draw [bend right=90, looseness=1.75] (0.center) to (1.center);
	\end{pgfonlayer}
\end{tikzpicture}$}}:= \begin{pmatrix}
1&0&0&1
\end{pmatrix}\qquad\qquad
\interp{\input{./figures/crossing.tikz}}:= \begin{pmatrix}
1&0&0&0\\
0&0&1&0\\
0&1&0&0\\
0&0&0&1
\end{pmatrix} \qquad\qquad
\interp{\raisebox{-0.35em}{$\begin{tikzpicture}
	\begin{pgfonlayer}{nodelayer}
		\node [style=none] (a0) at (-0.2500001, -0) {};
		\node [style=none] (a1) at (0.2500001, -0) {};
		\node [style=none] (a2) at (0, 0.25) {};
	\end{pgfonlayer}
	\begin{pgfonlayer}{edgelayer}
		\draw [bend left=90, looseness=1.75] (a0.center) to (a1.center);
	\end{pgfonlayer}
\end{tikzpicture}$}}:= \begin{pmatrix}
1\\0\\0\\1
\end{pmatrix}$$
$$
\interp{\begin{tikzpicture}
	\begin{pgfonlayer}{nodelayer}
	\node [style=gn] (0) at (0, -0) {$\alpha$};
	\end{pgfonlayer}
	\end{tikzpicture}}:=\begin{pmatrix}1+e^{i\alpha}\end{pmatrix}\qquad\qquad
\interp{\input{./figures/gn-alpha.tikz}}:=
\annoted{2^m}{2^n}{\begin{pmatrix}
	1 & 0 & \cdots & 0 & 0 \\
	0 & 0 & \cdots & 0 & 0 \\
	\vdots & \vdots & \ddots & \vdots & \vdots \\
	0 & 0 & \cdots & 0 & 0 \\
	0 & 0 & \cdots & 0 & e^{i\alpha}
	\end{pmatrix}}
~~\begin{pmatrix}n+m>0\end{pmatrix} 
$$
For any $n,m\geq 0$ and $\alpha\in\mathbb{R}$:\\
\begin{minipage}{\columnwidth}
	$$\scalebox{0.9}{$\interp{\input{./figures/rn-alpha.tikz}}=\interp{~\begin{tikzpicture}
	\begin{pgfonlayer}{nodelayer}
		\node [style=H box] (0) at (0, 0) {};
		\node [style=none] (1) at (0, 0.275) {};
		\node [style=none] (2) at (0, -0.275) {};
	\end{pgfonlayer}
	\begin{pgfonlayer}{edgelayer}
		\draw (2.center) to (1.center);
	\end{pgfonlayer}
\end{tikzpicture}
~}^{\otimes m}\circ \interp{\input{./figures/gn-alpha.tikz}}\circ \interp{~\begin{tikzpicture}
	\begin{pgfonlayer}{nodelayer}
		\node [style=H box] (0) at (0, 0) {};
		\node [style=none] (1) at (0, 0.275) {};
		\node [style=none] (2) at (0, -0.275) {};
	\end{pgfonlayer}
	\begin{pgfonlayer}{edgelayer}
		\draw (2.center) to (1.center);
	\end{pgfonlayer}
\end{tikzpicture}
~}^{\otimes n}$}$$ \\
	$\left(\text{where }M^{\otimes 0}=\begin{pmatrix}1\end{pmatrix}\text{ and }M^{\otimes k}=M\otimes M^{\otimes k-1}\text{ for }k\in \mathbb{N}^*\right)$.\\
	\rule{\columnwidth}{0.5pt}
\end{minipage}}\\


Theorem \ref{thm:ground} provides a recipe for transforming the language for mixed states and CPMs. The resulting language $\cdisc{\zx}$ can be seen as a prop with the generators of the ZX-Calculus, augmented with $\ground$ and with the axiomatisation enriched with $\{\ground\circ D = \ground~|~D^{\dagger}\circ D=I\}$. We actually do not need an infinite axiomatisation. Indeed, the set of isometries of the ZX-Calculus can be finitely generated.

Using ($e^{i\alpha}$, $\ket{0}$, H, $R_{Z}(\alpha)$, CNot) as spanning set of the isometries \cite{nielsen_chuang_2010}, we obtain only five axioms:
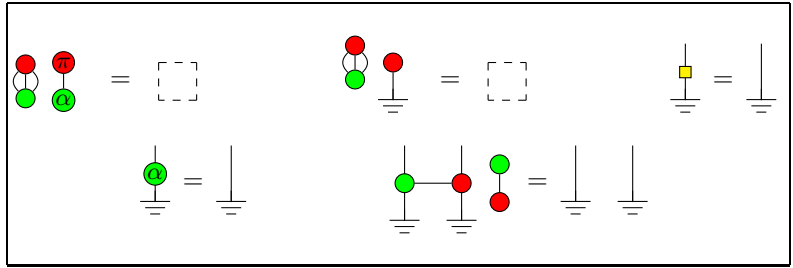
\begin{center}
	\begin{tabular}{|c|}
		\hline\\
		\input{./figures/ground-phase.tikz}$\qquad\qquad$
		\input{./figures/ground-ket-0.tikz}$\qquad\qquad$
		\input{./figures/ground-H-no-label.tikz}\\\\
 		\input{./figures/ground-gn-no-label.tikz}$\qquad\qquad$
		\input{./figures/ground-cnot-no-label.tikz}\\\\
		\hline
	\end{tabular}
\end{center}

\subsection{The $\frac{\pi}{2}$ fragment of ZX-calculus}

The $\textup{ZX}_{\frac \pi 2}$ is obtained from \zx by restricting
phases $\alpha$ to $\{0,\frac{\pi}{2},\pi,\frac{3\pi}{2}\}$. It is universal and complete for $\mathbf{Stab}$ \cite{pi_2-complete} with the axiomatisation provided in Figure \ref{fig:ZX_rules_clifford} in appendix. Moreover according to Lemma \ref{lem:CPS} $\cdisc{\mathbf{Stab}}$ is an environment structure for $\mathbf{Stab}$. 

The set ($e^{i\alpha}$, $\ket{0}$, $H$, $R_{Z}(\alpha)$, CNot), with
$\alpha$ restricted to multiples of $\frac{\pi}{2}$, remains a
spanning set of isometries in $\mathbf{Stab}$, so adding the same set
of equations than in $\cdisc{\zx}$ will provide a complete axiomatisation for $\cdisc{\textup{ZX}_{\frac \pi 2}}$.

\subsection{The $\mathbf{Clifford{+}T}$ fragment of ZX-calculus}

Restricting $ZX$ to angles multiples of $\pi/4$, we obtain a languages
which is known to be universal and complete for
$\mathbf{Clifford{+}T}$ \cite{JPV}. 
However, as shown by Lemma \ref{lem:CPC}, the semantic category
$\mathbf{Clifford{+}T}$ does not have enough isometries. The discard construction is strictly coarser than $\operatorname{CPM}$ for this fragment. So we leave open the complete axiomatisation of quantum operations for this fragment.

%

\subsection{The ZW-calculus}

The ZW-Calculus was introduced in \cite{zw}, deriving from the GHZ/W-Calculus \cite{ghz-w}, where the  main two generators are two non-equivalent ways to entangle three qubits, the so-called GHZ and W states. The language was made complete for pure quantum mechanics in \cite{HNW}. The generators, rules and interpretation of the calculus are given in the appendix at page \pageref{calc:ZW}. Since CNot is hard to express in this calculus, we choose another set of universal diagrams, more suited to ZW, namely ($e^{i\alpha}$, $\ket1$, $R_Z(\alpha)$, H, CZ $\circ$ SWAP). The resulting rules for $\cdisc{\text{ZW}}$ are:
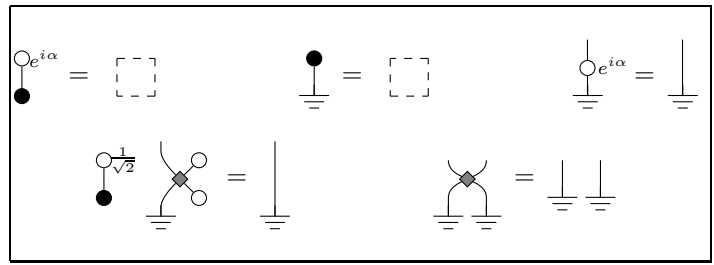
\begin{center}
	\begin{tabular}{|c|}
		\hline\\
		\input{./figures/ground-zw-phase.tikz}$\qquad\qquad$
		\input{./figures/ground-zw-ket-1.tikz}$\qquad\qquad$
		\input{./figures/ground-zw-ghz.tikz}\\\\
		\input{./figures/ground-zw-h.tikz}$\qquad\qquad$
		\input{./figures/ground-fermionic-swap.tikz}\\\\
		\hline
	\end{tabular}
\end{center}

\subsection{The ZH-Calculus}

\tikzstyle{gn}=[zxnode ,fill=white]
\tikzstyle{rn}=[zxnode ,fill=gray]
\tikzstyle{H box}=[rectangle,fill=white,draw=black,xscale=1,yscale=1,font=\footnotesize,inner sep=1.2pt,minimum width=0.15cm,minimum height=0.15cm]

The ZH-Calculus was introduced and proved to be complete in \cite{ZH}. A presentation of the language is given in appendix at page \pageref{calc:ZH}. The point of this language is to easily represent hypergraph-states, a generalisation of graph-states, a useful resource for quantum computing.
This language has been specifically designed to easily represent the multi-controlled Z (which constitute the hyperedges in the hypergraph-states). So in particular, CZ and $R_Z(\alpha)$ are easily representable. Up to a scalar, H is also easily doable, and $\interp{X^{(0,1)}}=\ket{0}$. Hence, choosing ($e^{i\alpha}$, $\ket{0}$, H, $R_{Z}(\alpha)$, CZ) as spanning set, we only need the axioms:
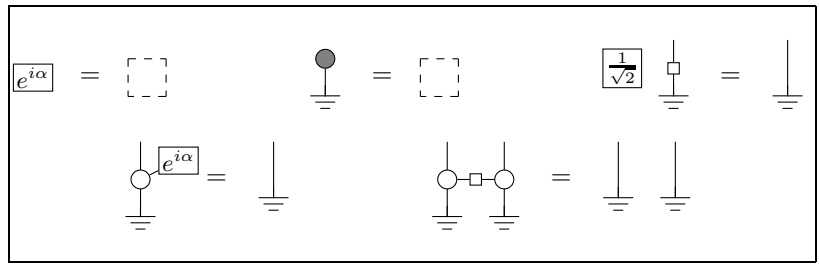
\begin{center}
	\begin{tabular}{|c|}
		\hline\\
		\input{./figures/ground-ZH-phase.tikz}$\qquad\qquad$
		\input{./figures/ground-ZH-ket-0.tikz}$\qquad\qquad$
		\input{./figures/ground-ZH-H.tikz}\\\\
		\input{./figures/ground-ZH-RZ.tikz}$\qquad\qquad$
		\input{./figures/ground-ZH-CZ.tikz}\\\\
		\hline
	\end{tabular}
\end{center}


\begin{thebibliography}{10}

\bibitem{cnot-dihedral}
Matthew Amy, Jianxin Chen, and Neil~J. Ross.
\newblock A finite presentation of cnot-dihedral operators.
\newblock In Bob Coecke and Aleks Kissinger, editors, {\em {\rm Proceedings
  14th International Conference on} Quantum Physics and Logic, {\rm Nijmegen,
  The Netherlands, 3-7 July 2017}}, volume 266 of {\em Electronic Proceedings
  in Theoretical Computer Science}, pages 84--97. Open Publishing Association,
  2018.
\newblock \href {http://dx.doi.org/10.4204/EPTCS.266.5}
  {\path{doi:10.4204/EPTCS.266.5}}.

\bibitem{clifford-state-partition}
Koenraad M~R Audenaert and Martin~B Plenio.
\newblock Entanglement on mixed stabilizer states: Normal forms and reduction
  procedures.
\newblock {\em New Journal of Physics}, 7:170--170, aug 2005.
\newblock URL: \url{https://doi.org/10.1088%2F1367-2630%2F7%2F1%2F170}, \href
  {http://dx.doi.org/10.1088/1367-2630/7/1/170}
  {\path{doi:10.1088/1367-2630/7/1/170}}.

\bibitem{pi_2-complete}
Miriam Backens.
\newblock The {ZX}-calculus is complete for stabilizer quantum mechanics.
\newblock {\em New Journal of Physics}, 16(9):093021, sep 2014.
\newblock URL: \url{https://doi.org/10.1088%2F1367-2630%2F16%2F9%2F093021},
  \href {http://dx.doi.org/10.1088/1367-2630/16/9/093021}
  {\path{doi:10.1088/1367-2630/16/9/093021}}.

\bibitem{pi_4-single-qubit}
Miriam Backens.
\newblock The {ZX}-calculus is complete for the single-qubit {C}lifford+{T}
  group.
\newblock {\em Electronic Proceedings in Theoretical Computer Science},
  172:293--303, dec 2014.
\newblock URL: \url{https://doi.org/10.4204%2Feptcs.172.21}, \href
  {http://dx.doi.org/10.4204/eptcs.172.21} {\path{doi:10.4204/eptcs.172.21}}.

\bibitem{toy-model-graph}
Miriam Backens and Ali~Nabi Duman.
\newblock A complete graphical calculus for {S}pekkens' toy bit theory.
\newblock {\em Foundations of Physics}, pages 1--34, 2014.
\newblock \href {http://arxiv.org/abs/arXiv:1411.1618}
  {\path{arXiv:arXiv:1411.1618}}, \href
  {http://dx.doi.org/10.1007/s10701-015-9957-7}
  {\path{doi:10.1007/s10701-015-9957-7}}.

\bibitem{ZH}
Miriam Backens and Aleks Kissinger.
\newblock Zh: A complete graphical calculus for quantum computations involving
  classical non-linearity.
\newblock In Peter Selinger and Giulio Chiribella, editors, {\em {\rm
  Proceedings of the 15th International Conference on} Quantum Physics and
  Logic, {\rm Halifax, Canada, 3-7th June 2018}}, volume 287 of {\em Electronic
  Proceedings in Theoretical Computer Science}, pages 23--42. Open Publishing
  Association, 2019.
\newblock \href {http://dx.doi.org/10.4204/EPTCS.287.2}
  {\path{doi:10.4204/EPTCS.287.2}}.

\bibitem{baez-PROP}
John~C. {Baez}, Brandon {Coya}, and Franciscus {Rebro}.
\newblock Props in network theory.
\newblock In {\em Theory and Applications of Categories}, volume~33, pages
  727--783, Jul 2017.
\newblock URL: \url{http://arxiv.org/abs/1707.08321}.

\bibitem{chancellor2016coherent}
Nicholas Chancellor, Aleks Kissinger, Joschka Roffe, Stefan Zohren, and Dominic
  Horsman.
\newblock Graphical structures for design and verification of quantum error
  correction.
\newblock last revised Jan. 2018, 2016.
\newblock URL: \url{https://arxiv.org/abs/1611.08012}.

\bibitem{mixed-states-axioms}
Bob Coecke.
\newblock Axiomatic description of mixed states from {S}elinger's
  {CPM}-construction.
\newblock {\em Electronic Notes in Theoretical Computer Science}, 210:3 -- 13,
  2008.
\newblock Proceedings of the 4th International Workshop on Quantum Programming
  Languages (QPL 2006).
\newblock URL:
  \url{http://www.sciencedirect.com/science/article/pii/S1571066108002296},
  \href {http://dx.doi.org/https://doi.org/10.1016/j.entcs.2008.04.014}
  {\path{doi:https://doi.org/10.1016/j.entcs.2008.04.014}}.

\bibitem{interacting}
Bob Coecke and Ross Duncan.
\newblock Interacting quantum observables: Categorical algebra and
  diagrammatics.
\newblock {\em New Journal of Physics}, 13(4):043016, apr 2011.
\newblock URL: \url{https://doi.org/10.1088%2F1367-2630%2F13%2F4%2F043016},
  \href {http://dx.doi.org/10.1088/1367-2630/13/4/043016}
  {\path{doi:10.1088/1367-2630/13/4/043016}}.

\bibitem{coecke2016pictures}
Bob Coecke and Chris Heunen.
\newblock Pictures of complete positivity in arbitrary dimension.
\newblock {\em Information and Computation}, 250:50--58, 2016.

\bibitem{ghz-w}
Bob Coecke and Aleks Kissinger.
\newblock The compositional structure of multipartite quantum entanglement.
\newblock In {\em Automata, Languages and Programming}, pages 297--308.
  Springer Berlin Heidelberg, 2010.
\newblock URL: \url{https://doi.org/10.1007%2F978-3-642-14162-1_25}, \href
  {http://dx.doi.org/10.1007/978-3-642-14162-1\_25}
  {\path{doi:10.1007/978-3-642-14162-1\_25}}.

\bibitem{picturing-qp}
Bob Coecke and Aleks Kissinger.
\newblock {\em Picturing Quantum Processes: A First Course in Quantum Theory
  and Diagrammatic Reasoning}.
\newblock Cambridge University Press, 2017.
\newblock \href {http://dx.doi.org/10.1017/9781316219317}
  {\path{doi:10.1017/9781316219317}}.

\bibitem{environment}
Bob Coecke and Simon Perdrix.
\newblock {Environment and Classical Channels in Categorical Quantum
  Mechanics}.
\newblock {\em {Logical Methods in Computer Science}}, {Volume 8, Issue 4},
  November 2012.
\newblock URL: \url{https://lmcs.episciences.org/719}, \href
  {http://dx.doi.org/10.2168/LMCS-8(4:14)2012}
  {\path{doi:10.2168/LMCS-8(4:14)2012}}.

\bibitem{2-qubits-zx}
Bob Coecke and Quanlong Wang.
\newblock {ZX}-rules for 2-qubit {C}lifford+{T} quantum circuits, 2018.
\newblock \href {http://arxiv.org/abs/1804.05356} {\path{arXiv:1804.05356}}.

\bibitem{de2017zx}
Niel de~Beaudrap and Dominic Horsman.
\newblock The {ZX}-calculus is a language for surface code lattice surgery.
\newblock {\em CoRR}, abs/1704.08670, 2017.
\newblock URL: \url{http://arxiv.org/abs/1704.08670}, \href
  {http://arxiv.org/abs/1704.08670} {\path{arXiv:1704.08670}}.

\bibitem{duncan2013mbqc}
Ross Duncan.
\newblock A graphical approach to measurement-based quantum computing.
\newblock In {\em Quantum Physics and Linguistics}, pages 50--89. Oxford
  University Press, feb 2013.
\newblock URL:
  \url{https://doi.org/10.1093%2Facprof%3Aoso%2F9780199646296.003.0003}, \href
  {http://dx.doi.org/10.1093/acprof:oso/9780199646296.003.0003}
  {\path{doi:10.1093/acprof:oso/9780199646296.003.0003}}.

\bibitem{duncan2016hopf}
Ross Duncan and Kevin Dunne.
\newblock Interacting {F}robenius algebras are {H}opf.
\newblock In {\em Proceedings of the 31st Annual ACM/IEEE Symposium on Logic in
  Computer Science}, LICS 2016, pages 535--544, New York, NY, USA, 2016. ACM.
\newblock URL: \url{http://doi.acm.org/10.1145/2933575.2934550}, \href
  {http://dx.doi.org/10.1145/2933575.2934550}
  {\path{doi:10.1145/2933575.2934550}}.

\bibitem{verifying-color-code}
Ross Duncan and Liam Garvie.
\newblock Verifying the smallest interesting colour code with quantomatic.
\newblock In Bob Coecke and Aleks Kissinger, editors, {\em {\rm Proceedings
  14th International Conference on} Quantum Physics and Logic, {\rm Nijmegen,
  The Netherlands, 3-7 July 2017}}, volume 266 of {\em Electronic Proceedings
  in Theoretical Computer Science}, pages 147--163. Open Publishing
  Association, 2018.
\newblock \href {http://dx.doi.org/10.4204/EPTCS.266.10}
  {\path{doi:10.4204/EPTCS.266.10}}.

\bibitem{duncan2019graph}
Ross Duncan, Aleks Kissinger, Simon Perdrix, and John van~de Wetering.
\newblock Graph-theoretic simplification of quantum circuits with the
  zx-calculus.
\newblock {\em arXiv preprint arXiv:1902.03178}, 2019.

\bibitem{duncan2014verifying}
Ross Duncan and Maxime Lucas.
\newblock Verifying the {S}teane code with {Q}uantomatic.
\newblock {\em Electronic Proceedings in Theoretical Computer Science},
  171:33--49, dec 2014.
\newblock URL: \url{https://doi.org/10.4204%2Feptcs.171.4}, \href
  {http://dx.doi.org/10.4204/eptcs.171.4} {\path{doi:10.4204/eptcs.171.4}}.

\bibitem{mbqc}
Ross Duncan and Simon Perdrix.
\newblock Rewriting measurement-based quantum computations with generalised
  flow.
\newblock {\em Lecture Notes in Computer Science}, 6199:285--296, 2010.
\newblock URL: \url{http://personal.strath.ac.uk/ross.duncan/papers/gflow.pdf},
  \href {http://dx.doi.org/10.1007/978-3-642-14162-1\_24}
  {\path{doi:10.1007/978-3-642-14162-1\_24}}.

\bibitem{pivoting}
Ross Duncan and Simon Perdrix.
\newblock Pivoting makes the {ZX}-calculus complete for real stabilizers.
\newblock In {\em QPL 2013}, Electronic Proceedings in Theoretical Computer
  Science, pages 50--62, 2013.
\newblock \href {http://arxiv.org/abs/arXiv:1307.7048}
  {\path{arXiv:arXiv:1307.7048}}, \href {http://dx.doi.org/10.4204/EPTCS.171.5}
  {\path{doi:10.4204/EPTCS.171.5}}.

\bibitem{zw}
Amar Hadzihasanovic.
\newblock A diagrammatic axiomatisation for qubit entanglement.
\newblock In {\em 2015 30th Annual ACM/IEEE Symposium on Logic in Computer
  Science}, pages 573--584, July 2015.
\newblock \href {http://dx.doi.org/10.1109/LICS.2015.59}
  {\path{doi:10.1109/LICS.2015.59}}.

\bibitem{Amar}
Amar Hadzihasanovic.
\newblock {\em The Algebra of Entanglement and the Geometry of Composition}.
\newblock PhD thesis, University of Oxford, 2017.
\newblock URL: \url{https://arxiv.org/abs/1709.08086}.

\bibitem{HNW}
Amar Hadzihasanovic, Kang~Feng Ng, and Quanlong Wang.
\newblock Two complete axiomatisations of pure-state qubit quantum computing.
\newblock In {\em Proceedings of the 33rd Annual ACM/IEEE Symposium on Logic in
  Computer Science}, LICS '18, pages 502--511, New York, NY, USA, 2018. ACM.
\newblock URL: \url{http://doi.acm.org/10.1145/3209108.3209128}, \href
  {http://dx.doi.org/10.1145/3209108.3209128}
  {\path{doi:10.1145/3209108.3209128}}.

\bibitem{horsman2011quantum}
Clare Horsman.
\newblock Quantum picturalism for topological cluster-state computing.
\newblock {\em New Journal of Physics}, 13(9):095011, sep 2011.
\newblock URL: \url{https://doi.org/10.1088%2F1367-2630%2F13%2F9%2F095011},
  \href {http://dx.doi.org/10.1088/1367-2630/13/9/095011}
  {\path{doi:10.1088/1367-2630/13/9/095011}}.

\bibitem{huotuniversal}
Mathieu Huot and Sam Staton.
\newblock Universal properties in quantum theory.
\newblock In Peter Selinger and Giulio Chiribella, editors, {\em {\rm
  Proceedings of the 15th International Conference on} Quantum Physics and
  Logic, {\rm Halifax, Canada, 3-7th June 2018}}, volume 287 of {\em Electronic
  Proceedings in Theoretical Computer Science}, pages 213--223. Open Publishing
  Association, 2019.
\newblock \href {http://dx.doi.org/10.4204/EPTCS.287.12}
  {\path{doi:10.4204/EPTCS.287.12}}.

\bibitem{JPV}
Emmanuel Jeandel, Simon Perdrix, and Renaud Vilmart.
\newblock A complete axiomatisation of the {ZX}-calculus for {C}lifford+{T}
  quantum mechanics.
\newblock In {\em Proceedings of the 33rd Annual ACM/IEEE Symposium on Logic in
  Computer Science}, LICS '18, pages 559--568, New York, NY, USA, 2018. ACM.
\newblock URL: \url{http://doi.acm.org/10.1145/3209108.3209131}, \href
  {http://dx.doi.org/10.1145/3209108.3209131}
  {\path{doi:10.1145/3209108.3209131}}.

\bibitem{JPV-universal}
Emmanuel Jeandel, Simon Perdrix, and Renaud Vilmart.
\newblock Diagrammatic reasoning beyond {C}lifford+{T} quantum mechanics.
\newblock In {\em Proceedings of the 33rd Annual ACM/IEEE Symposium on Logic in
  Computer Science}, LICS '18, pages 569--578, New York, NY, USA, 2018. ACM.
\newblock URL: \url{http://doi.acm.org/10.1145/3209108.3209139}, \href
  {http://dx.doi.org/10.1145/3209108.3209139}
  {\path{doi:10.1145/3209108.3209139}}.

\bibitem{ZXNormalForm}
Emmanuel Jeandel, Simon Perdrix, and Renaud Vilmart.
\newblock A generic normal form for zx-diagrams and application to the rational
  angle completeness.
\newblock 2018.
\newblock \href {http://arxiv.org/abs/1805.05296} {\path{arXiv:1805.05296}}.

\bibitem{pyzx}
A.~Kissinger and John van~de Wetering.
\newblock Pyzx, 2018.
\newblock URL: \url{https://github.com/Quantomatic/pyzx}.

\bibitem{kissinger2017categorical}
Aleks Kissinger and Sander Uijlen.
\newblock A categorical semantics for causal structure.
\newblock In {\em 2017 32nd Annual ACM/IEEE Symposium on Logic in Computer
  Science (LICS)}, pages 1--12. IEEE, 2017.

\bibitem{kissinger2015quantomatic}
Aleks Kissinger and Vladimir Zamdzhiev.
\newblock Quantomatic: A proof assistant for diagrammatic reasoning.
\newblock In Amy~P. Felty and Aart Middeldorp, editors, {\em Automated
  Deduction - CADE-25}, pages 326--336, Cham, 2015. Springer International
  Publishing.
\newblock \href {http://dx.doi.org/10.1007/978-3-319-21401-6\_22}
  {\path{doi:10.1007/978-3-319-21401-6\_22}}.

\bibitem{matsumoto-amano}
Ken Matsumoto and Kazuyuki Amano.
\newblock {Representation of Quantum Circuits with Clifford and $\pi/8$ Gates},
  June 2008.
\newblock \href {http://arxiv.org/abs/0806.3834} {\path{arXiv:0806.3834}}.

\bibitem{nielsen_chuang_2010}
Michael~A. Nielsen and Isaac~L. Chuang.
\newblock {\em Quantum Computation and Quantum Information: 10th Anniversary
  Edition}.
\newblock Cambridge University Press, 2010.
\newblock \href {http://dx.doi.org/10.1017/CBO9780511976667}
  {\path{doi:10.1017/CBO9780511976667}}.

\bibitem{Selinger:2004:TQP:1014327.1014330}
Peter Selinger.
\newblock Towards a quantum programming language.
\newblock {\em Mathematical. Structures in Comp. Sci.}, 14(4):527--586, August
  2004.
\newblock URL: \url{https://doi.org/10.1017/S0960129504004256}, \href
  {http://dx.doi.org/10.1017/S0960129504004256}
  {\path{doi:10.1017/S0960129504004256}}.

\bibitem{Selinger-CPM}
Peter Selinger.
\newblock Dagger compact closed categories and completely positive maps.
\newblock {\em Electronic Notes in Theoretical Computer Science}, 170:139--163,
  mar 2007.
\newblock URL: \url{https://doi.org/10.1016%2Fj.entcs.2006.12.018}, \href
  {http://dx.doi.org/10.1016/j.entcs.2006.12.018}
  {\path{doi:10.1016/j.entcs.2006.12.018}}.

\bibitem{selinger2010survey}
Peter Selinger.
\newblock A survey of graphical languages for monoidal categories.
\newblock In {\em New structures for physics}, pages 289--355. Springer, 2010.

\bibitem{clifford-circuits}
Peter Selinger.
\newblock {Generators and Relations for n-qubit {C}lifford Operators}.
\newblock {\em {Logical Methods in Computer Science}}, {Volume 11, Issue 2},
  June 2015.
\newblock URL: \url{https://lmcs.episciences.org/1570}, \href
  {http://dx.doi.org/10.2168/LMCS-11(2:10)2015}
  {\path{doi:10.2168/LMCS-11(2:10)2015}}.

\bibitem{SelingerBian}
Peter Selinger and Xiaoning Bian.
\newblock Relations for {C}lifford+{T} operators on two qubits, 2015.
\newblock URL: \url{https://www.mathstat.dal.ca/~xbian/talks/}.

\bibitem{euler-zx}
Renaud Vilmart.
\newblock A near-optimal axiomatisation of {ZX}-calculus for pure qubit quantum
  mechanics.
\newblock 2018.
\newblock URL: \url{https://arxiv.org/abs/1812.09114}, \href
  {http://arxiv.org/abs/arXiv:1812.09114} {\path{arXiv:arXiv:1812.09114}}.

\bibitem{PhD.Zanasi}
Fabio Zanasi.
\newblock {\em Interacting Hopf Algebras -- the theory of linear systems}.
\newblock PhD thesis, Universit\'e de Lyon, 2015.
\newblock URL: \url{http://www.zanasi.com/fabio/#/publications.html}.

\end{thebibliography}

\appendix

\section{Proofs}

\begin{proof}[Proof of Lemma \ref{lem:pur}]\phantomsection\label{prf:pur}
	Given any morphism $f:\mathbf{C}^{\sground}[A,B]$, we take a diagram representing it. Using the naturality of the symmetry we obtain an equivalent diagram in $\cdisc{\mathbf{C}}$ where all the discards have been pushed to the bottom right: \input{./figures/funground.tikz}.
	There are no discards among the components of the part $f''$ of this diagram. So it represents a morphism in the range of $\iota_{\mathbf{C}}$ and then there is an $f':\mathbf{C}[A,B\otimes X]$ such that:
	\begin{equation*}
	\input{./figures/ffig0.tikz}=\input{./figures/ffig1.tikz}
	\end{equation*}
	In other words, $f'$ is a purification of $f$.
\end{proof}

\begin{proof}[Proof of Lemma~\ref{lem:unicity}]\phantomsection\label{prf:unicity}

\phantom{a}
  
	\begin{itemize}
		
		\item[$(\Rightarrow)$] It is enough to show $f\iso g \Rightarrow \input{./figures/lm0.tikz}=\input{./figures/lm1.tikz}$ since equality is transitive.
		
		$f\iso g \Leftrightarrow$ there are two isometries $u:X\to Z$ and $v:Y\to Z$ such that $\input{./figures/simis0.tikz}=\input{./figures/simis1.tikz}$ and then:
		
		\begin{align*}
		&\input{./figures/simis0.tikz}=\input{./figures/simis1.tikz} \Rightarrow \input{./figures/iosimis0.tikz}=\input{./figures/iosimis1.tikz}\\ &\Rightarrow \input{./figures/giosimis0.tikz}=\input{./figures/giosimis1.tikz} \Rightarrow \input{./figures/lm0.tikz}=\input{./figures/lm1.tikz}
		\end{align*}
		
		\item[$(\Leftarrow)$] We have
                  $\input{./figures/lm0.tikz}=\input{./figures/lm1.tikz}$ in
                  $\mathbf{C}^{\sground}$.
                  To do the proof, we will have to go back to the
                  definition of the category $\mathbf{C}^{\sground}$
                  as a pushout.
                  Recall that two terms are equal if one can
                  rewrite one into the other  using the equations
                  defining $\cdisc{\mathbf{C}}$.

                  We can assume that, among those steps, the only one
                  involving discards are isometry
                  deletion/creation. Diagramatically this amounts to say that
                  the discards are never moved, in fact one can always
                  moves the other morphisms to make them interact with
                  the discards.

                  Doing this, we ensure that all intermediary diagrams in the chain
                  of equations are of the form $\input{./figures/pk.tikz}$ for some
                  $k$.
                  Therefore, to prove the result for a chain of
                  equations of arbitrary size, it is enough to do it just
                  for one step of rewriting.

                  Consider then this step of rewriting. There are two
                  cases. Either we have used an equation which, by
                  identification, can be seen as an equation of
                  $\mathbf{C}$, that is which involves no discards.
                  Then by functoriality of $\iota_{C}$ we recover that
                  $f=g$ and therefore $f\iso g$.                
                  Or the equation involves a discard which has deleted an isometry $u$. Then one of the upper part, let's say $\iota_{C}(f)$, can be written 
		\def\fig{simis}$\begin{tikzpicture}[scale=0.5]
	\begin{pgfonlayer}{nodelayer}
		\node [style=none] (0)  at (-1.0, 0.75) {};
		\node [style=none] (1)  at (1.0, 0.75) {};
		\node [style=none] (2)  at (-1.0, -0.25) {};
		\node [style=none] (3)  at (1.0, -0.25) {};
		\node [style=none] (4)  at (-0.5, -0.25) {};
		\node [style=none] (5)  at (-0.5, -1.5) {};
		\node [style=none] (6)  at (0.0, 0.75) {};
		\node [style=none] (7)  at (0.0, 1.5) {};
		\node [style=none] (8)  at (0.0, 0.25) {$\iota_{\textbf{C}}(f)$};
		\node [style=none] (9)  at (0.5, -0.25) {};
		\node [style=none] (16)  at (0.5, -1.5) {};
	\end{pgfonlayer}
	\begin{pgfonlayer}{edgelayer}
		\draw (0.center) to (1.center);
		\draw (1.center) to (3.center);
		\draw (2.center) to (0.center);
		\draw (3.center) to (2.center);
		\draw (4.center) to (5.center);
		\draw (7.center) to (6.center);
		\draw (9.center) to (16.center);
	\end{pgfonlayer}
\end{tikzpicture}\eq{}\begin{tikzpicture}[scale=0.5]
	\begin{pgfonlayer}{nodelayer}
		\node [style=none] (19)  at (-1.0, 1.0) {};
		\node [style=none] (20)  at (1.0, 1.0) {};
		\node [style=none] (21)  at (-1.0, 0.0) {};
		\node [style=none] (22)  at (1.0, 0.0) {};
		\node [style=none] (23)  at (-0.5, 0.0) {};
		\node [style=none] (24)  at (-0.5, -1.5) {};
		\node [style=none] (25)  at (0.0, 1.0) {};
		\node [style=none] (26)  at (0.0, 1.5) {};
		\node [style=none] (27)  at (0.0, 0.5) {$\iota_{\textbf{C}}(g)$};
		\node [style=none] (28)  at (0.5, 0.0) {};
		\node [style=none] (29)  at (0.0, -0.25) {};
		\node [style=none] (30)  at (1.0, -0.25) {};
		\node [style=none] (31)  at (1.0, -1.25) {};
		\node [style=none] (32)  at (0.0, -1.25) {};
		\node [style=none] (33)  at (0.5, -1.25) {};
		\node [style=none] (34)  at (0.5, -1.5) {};
		\node [style=none] (35)  at (0.5, -0.25) {};
		\node [style=none] (36)  at (0.5, -0.75) {$u$};
	\end{pgfonlayer}
	\begin{pgfonlayer}{edgelayer}
		\draw (19.center) to (20.center);
		\draw (20.center) to (22.center);
		\draw (21.center) to (19.center);
		\draw (22.center) to (21.center);
		\draw (23.center) to (24.center);
		\draw (26.center) to (25.center);
		\draw (28.center) to (35.center);
		\draw (29.center) to (30.center);
		\draw (30.center) to (31.center);
		\draw (31.center) to (32.center);
		\draw (32.center) to (29.center);
		\draw (33.center) to (34.center);
	\end{pgfonlayer}
\end{tikzpicture}$.
		But $u$ being an isometry, there exists $u'$ in $\mathbf{C}$ such that $\iota_{\mathbf{C}}(u')=u$. Hence, we have 
		$\begin{tikzpicture}[scale=0.5]
	\begin{pgfonlayer}{nodelayer}
		\node [style=none] (38)  at (-1.0, 0.75) {};
		\node [style=none] (39)  at (1.0, 0.75) {};
		\node [style=none] (40)  at (-1.0, -0.25) {};
		\node [style=none] (41)  at (1.0, -0.25) {};
		\node [style=none] (42)  at (-0.5, -0.25) {};
		\node [style=none] (43)  at (-0.5, -1.5) {};
		\node [style=none] (44)  at (0.0, 0.75) {};
		\node [style=none] (45)  at (0.0, 1.5) {};
		\node [style=none] (46)  at (0.0, 0.25) {$f$};
		\node [style=none] (47)  at (0.5, -0.25) {};
		\node [style=none] (48)  at (0.5, -1.5) {};
	\end{pgfonlayer}
	\begin{pgfonlayer}{edgelayer}
		\draw (38.center) to (39.center);
		\draw (39.center) to (41.center);
		\draw (40.center) to (38.center);
		\draw (41.center) to (40.center);
		\draw (42.center) to (43.center);
		\draw (45.center) to (44.center);
		\draw (47.center) to (48.center);
	\end{pgfonlayer}
\end{tikzpicture}\eq{}\begin{tikzpicture}[scale=0.5]
	\begin{pgfonlayer}{nodelayer}
		\node [style=none] (50)  at (-1.0, 1.0) {};
		\node [style=none] (51)  at (1.0, 1.0) {};
		\node [style=none] (52)  at (-1.0, 0.0) {};
		\node [style=none] (53)  at (1.0, 0.0) {};
		\node [style=none] (54)  at (-0.5, 0.0) {};
		\node [style=none] (55)  at (-0.5, -1.5) {};
		\node [style=none] (56)  at (0.0, 1.0) {};
		\node [style=none] (57)  at (0.0, 1.5) {};
		\node [style=none] (58)  at (0.0, 0.5) {$g$};
		\node [style=none] (59)  at (0.5, 0.0) {};
		\node [style=none] (60)  at (0.0, -0.25) {};
		\node [style=none] (61)  at (1.0, -0.25) {};
		\node [style=none] (62)  at (1.0, -1.25) {};
		\node [style=none] (63)  at (0.0, -1.25) {};
		\node [style=none] (64)  at (0.5, -1.25) {};
		\node [style=none] (65)  at (0.5, -1.5) {};
		\node [style=none] (66)  at (0.5, -0.25) {};
		\node [style=none] (67)  at (0.5, -0.75) {$u'$};
	\end{pgfonlayer}
	\begin{pgfonlayer}{edgelayer}
		\draw (50.center) to (51.center);
		\draw (51.center) to (53.center);
		\draw (52.center) to (50.center);
		\draw (53.center) to (52.center);
		\draw (54.center) to (55.center);
		\draw (57.center) to (56.center);
		\draw (59.center) to (66.center);
		\draw (60.center) to (61.center);
		\draw (61.center) to (62.center);
		\draw (62.center) to (63.center);
		\draw (63.center) to (60.center);
		\draw (64.center) to (65.center);
	\end{pgfonlayer}
\end{tikzpicture}$ in $\mathbf{C}$. It follows that $f\iso g$.
	\end{itemize}
	
\end{proof}

\begin{proof}[Proof of Lemma \ref{lem:inclu}]\phantomsection\label{prf:inclu}
	Since $\cp  $ is transitive it is enough to show that $\iso $ $\subseteq$ $\cp  $.
	Let $f: A \rightarrow B\otimes X$ and $g:A \rightarrow
        B\otimes Y$ s.t. 	$f\iso g$.	
Then there are two isometries $u:X\to Z$ and $v:Y\to Z$ such that $\input{./figures/simis0.tikz}=\input{./figures/simis1.tikz}$ and then:
	
	\begin{equation*}
	\input{./figures/inc__0.tikz}\eq{}\input{./figures/inc__1.tikz}\eq{}\input{./figures/inc__2.tikz}\eq{}\input{./figures/inc__3.tikz}
	\end{equation*} 
	
	So $f\cp  g$.
\end{proof}

\begin{proof}[Proof of Lemma \ref{thm:cpenv}]\phantomsection\label{prf:cpenv}

	[$(i)\Leftrightarrow (ii)$] First $\mathbf{C}^{\sground}$ has the same object as $\mathbf{C}$ and $\iota_{\mathbf{C}}:\mathbf{C}\to \overline{\mathbf{C}}$ is a SM-functor. We need to check the three conditions hold:
	\begin{itemize}
		\item[$\bullet$] Since $\iota_{\mathbf{C}_{iso}^{\term{}}}$ is strict monoidal one has:
		\begin{align*}
		\overset{I}{\ground}&=\iota_{\mathbf{C}_{iso}^{\term{}}}(\term{I})=\iota_{\mathbf{C}_{iso}^{\term{}}}(id_{I})=id_{I}\\
		\overset{A}{\ground}\otimes\overset{B}{\ground}&=\iota_{\mathbf{C}_{iso}^{\term{}}}(\term{A})\otimes\iota_{\mathbf{C}_{iso}^{\term{}}}(\term{B})=\iota_{\mathbf{C}_{iso}^{\term{}}}(\term{A}\otimes\term{B})\\&=\iota_{\mathbf{C}_{iso}^{\term{}}}(\term{A\otimes B})=\overset{A\otimes B}{\ground}
		\end{align*}
		So the first condition is satisfied.\\
		\item[$\bullet$] The second condition is Lemma \ref{lem:pur}.
		\item[$\bullet$] According to Lemma \ref{lem:inclu}, $\isop$ $\subseteq$ $\cp  $, thus the third condition is satisfied if and only if $\cp  $ $\subseteq$ $\isop$.  
	\end{itemize}
	
	[$(i)\Leftrightarrow (ii)$] Direct consequence of the fact that $\mathbb D$ is an environment structure for $\mathbb C$ iff $D$ is equivalent to $CPM(\mathbb C)$ \cite{coecke2016pictures}. 
	\end{proof}

\begin{proof}[Proof of Proposition \ref{lem:CPFH}]\phantomsection\label{prf:CPFH}
	Let $f:A\to B \otimes X$ and $g:A\to B\otimes Y$ be two linear maps such that $f\cp g$. By definition: $\input{./figures/simcpm0.tikz}=\input{./figures/simcpm1.tikz}$. It follows that the two superoperators $\rho\mapsto \tr_{X}(f^{\dagger}\rho f)$ and $\rho\mapsto \tr_{Y}(g^{\dagger}\rho g)$ are equal and then by the Stinespring dilation theorem (see for example \cite{huotuniversal}), there are isometries $u$ and $v$ such that $\input{./figures/simis0.tikz}=\input{./figures/simis1.tikz}$. In other words $f\iso g$. This shows that $\cp \subseteq \iso$ which is even stronger than the CP-condition. From Lemma \ref{lem:inclu} it follows that $\isop \subseteq \iso$. 
\end{proof}

\begin{proof}[Proof of Proposition \ref{lem:CPQ}]\label{prf:CPQ}
It suffices to remark that, in the preceding proof for FHilb, we might
suppose wlog that $u$ and $v$ have codomain of the form
$\mathbb{C}^{2^n}$, by postcomposing them if necessary with an
isometry from $\mathbb{C}^m$ to $\mathbb{C}^{2^n}$.

Therefore $f \cp g$ on $\mathbf{Qubit}$ implies $f \iso g$ on $\mathbf{Qubit}$.
  
\end{proof}
\begin{proof}[Proof of Proposition \ref{lem:CPS}]\phantomsection\label{prf:CPS}
	First of all, since $\mathbf{Stab}$ is compact closed, using the map/state duality, proving the result for states in sufficient. Since all the non-zero scalar are invertible in $\mathbf{Stab}$ we can furthermore without loss of generality focusing on normalized states.  
	Consider two states $d_{1}:A\otimes X$ and $d_{2}:A\otimes Y$ in $\mathbf{Stab}$ such that $d_1 \cp d_2$. The point of focusing on normalized states is that we can decompose using \cite{clifford-state-partition} so that $\input{./figures/di.tikz}=\input{./figures/decstab.tikz}$ where $A_i$ and $B_i$ are unitaries in $\mathbf{Stab}$. Defining $A_{i}'\eqqcolon\input{./figures/ai.tikz}$ we have that $d_{i}\iso A_{i}'$ since we just have deleted isometries. So, by transitivity, to prove $d_{1}\isop d_{2}$ we just have to show $A_{1}'\iso A_{2}'$. But since $d_1 \cp d_2$ in $\mathbf{Stab}$ we also have $d_1 \cp d_2$ in $\mathbf{FHilb}$ and so by Lemma \ref{lem:CPFH}, $d_1 \isop d_2$ in $\mathbf{FHilb}$. By transitivity $A_{1}'\isop A_{2}'$ in $\mathbf{FHilb}$ and so by Lemma \ref{lem:CPFH} $A_{1}'\iso A_{1}'$ in $\mathbf{FHilb}$. So there are two unitaries $u$ and $v$ such that $\input{./figures/a1u.tikz}=\input{./figures/a2v.tikz}$. In $\mathbf{FHilb}$ any isometry can be written as an unitary with ancillas. In other words there is an unitary $u'$ such that: $\input{./figures/u.tikz}=\input{./figures/u_.tikz}$, composing by $u'^{\dagger}$ on both side and denoting $w=u'^\dagger\circ v$ one has: $\input{./figures/a1.tikz}=\input{./figures/a2w.tikz}$. It only remains to show that the isometry $w$ is in $\mathbf{Stab}$ since the isometry on left hand side is clearly in it. This is given by: $\input{./figures/a1_.tikz}=\input{./figures/w.tikz}$ so $A_{1}'\iso A_{1}'$ in $\mathbf{Stab}$ and then $d_1 \cp d_2$.
\end{proof}

\begin{proof}[Proof of Proposition \ref{lem:CPC}]\phantomsection\label{prf:CPC}
First remark that, in any $\dagger$-SMC category, if  $f\isop g$ then there is a
morphism (usually not an isometry) $w$ such that $\input{./figures/f.tikz}=\input{./figures/gw.tikz}$.

This is true if $f \iso g$: From
$\input{./figures/simis0.tikz}=\input{./figures/simis1.tikz}$ we immediately get $\input{./figures/simis-inv.tikz}$.

The result then follows by a straightforward induction.

Now take $\phi = 1+2i$ and $\phi^* = 1 - 2i$.
The scalars are in $\mathbf{Clifford{+}T}$ since their entries are in
$\mathbb Z[i, \frac 1{\sqrt 2}]$, and are clearly $\cp$
equivalent. Now let's suppose $ 1+2i \isop 1-2i$.
Then by the previous remark, there exists a morphism $u$ such that
$(1-2i)u=1+2i$. But the only possibility for $u$ is $\frac{4i-3}{5}$,
which is not in $\mathbb Z[i, \frac 1{\sqrt 2}]$, a contradiction.
\end{proof}

\begin{proof}[Proof of Lemma \ref{lem:Fis}]\phantomsection\label{prf:Fis}
        First, remark that if $F(\ell) \iso k$, then there exists $h$
        s.t. $F(h) = k$. Indeed, under the hypothesis, there are two
        isometries $u$ and $v$ such that:
        $\input{./figures/simisf6.tikz}=\input{./figures/simisf7.tikz}$. Since $\ciso{F}$
        is surjective, there are two isometries $a$ and $b$ such that $F(a)=u$ and $f(b)=v$.
	\begin{equation*}
	\input{./figures/simisf8.tikz}=\input{./figures/simisf9.tikz} \Rightarrow \input{./figures/simisf10.tikz}=\input{./figures/simisf11.tikz} \Rightarrow F\left(\input{./figures/simisf12.tikz}\right)=\input{./figures/simisf11.tikz}
	\end{equation*}
	
	The first implication uses the fact that $F(b)$ is an
        isometry. So $k$ is
        in the image of $F$.

        By the first remark, it is therefore sufficient to prove the
        result if $F(f) \iso F(g)$.
	Since $\ciso{F}$ is surjective, there are two isometries $a$
        and $b$ such that $F(a)=u$ and $f(b)=v$. Therefore
	
	\begin{equation*}
	\input{./figures/simisf2.tikz}=\input{./figures/simisf3.tikz} \Rightarrow F\left(\input{./figures/simisf5.tikz}\right)=F\left(\input{./figures/simisf4.tikz} \right) \Rightarrow \input{./figures/simisf5.tikz}=\input{./figures/simisf4.tikz}
	\end{equation*}
	
The second one holds because $F$ is faithful. The last equation is the definition of $f\iso  g$.
\end{proof}

\begin{proof}[Proof of Theorem \ref{thm:ground}]\phantomsection\label{prf:ground}
	Let $f$ and $g$ be two morphisms such that $\cdisc{F}(f)=\cdisc{F}(g)$. By Lemma \ref{lem:pur}, $f$ and $g$ can be purified: 
	
	\begin{equation*}
	\cdisc{F}\left(\input{./figures/faithp0.tikz}\right)=\cdisc{F}\left(\input{./figures/faithp1.tikz}\right) \Rightarrow \input{./figures/faithp2.tikz}=\input{./figures/faithp3.tikz}
	\end{equation*}
	
	The implication follows from the upper face of the commutative cube.
	By Lemma \ref{lem:unicity} we have $F(f')\isop F(g')$. By Lemma \ref{lem:Fis}, $f'\isop g'$. Then Lemma \ref{lem:unicity} gives $\input{./figures/faithp0.tikz}=\input{./figures/faithp1.tikz}$ that is $f=g$, $F$ is faithful.
\end{proof}

\section{ZW and ZH Calculi}

\subsection{ZW-calculus}
\label{calc:ZW}

ZW-diagrams are generated by:
\begin{center}
	\bgroup
	\def\arraystretch{2.5}
	{\begin{tabular}{|cc|cc|}
			\hline
			$Z^{(n,m)}(r):n\to m$ & \vphantom{\scalebox{1.15}{\input{./figures/ZW-gen-GHZ-r.tikz}}}\input{./figures/ZW-gen-GHZ-r.tikz} & $\mathbb{I}:1\to 1$ & \begin{tikzpicture}
	\begin{pgfonlayer}{nodelayer}
		\node [style=none] (0) at (0, 0.2499999) {};
		\node [style=none] (1) at (0, -0.2499999) {};
	\end{pgfonlayer}
	\begin{pgfonlayer}{edgelayer}
		\draw (0.center) to (1.center);
	\end{pgfonlayer}
\end{tikzpicture}\\\hline
			$W^{(n,m)}:n\to m$ & \vphantom{\scalebox{1.15}{\input{./figures/ZW-gen-W.tikz}}}\input{./figures/ZW-gen-W.tikz} & $e:0\to 0$ & \input{./figures/empty-diagram.tikz}\\\hline
			$\sigma:2\to 2$ & \input{./figures/crossing.tikz} & $\sigma':2\to 2$ & \input{./figures/zw-cross.tikz}\\\hline
			$\epsilon:2\to 0$ & \begin{tikzpicture}
	\begin{pgfonlayer}{nodelayer}
		\node [style=none] (0) at (-0.2500001, 0.2500001) {};
		\node [style=none] (1) at (0.2500001, 0.2500001) {};
	\end{pgfonlayer}
	\begin{pgfonlayer}{edgelayer}
		\draw [bend right=90, looseness=1.75] (0.center) to (1.center);
	\end{pgfonlayer}
\end{tikzpicture} & $\eta:0\to 2$ & \begin{tikzpicture}
	\begin{pgfonlayer}{nodelayer}
		\node [style=none] (a0) at (-0.2500001, -0) {};
		\node [style=none] (a1) at (0.2500001, -0) {};
		\node [style=none] (a2) at (0, 0.25) {};
	\end{pgfonlayer}
	\begin{pgfonlayer}{edgelayer}
		\draw [bend left=90, looseness=1.75] (a0.center) to (a1.center);
	\end{pgfonlayer}
\end{tikzpicture}\\\hline
	\end{tabular}}
	\egroup\\
	where $n,m\in \mathbb{N}$, $r \in \mathbb{C}$, and the generator $e$ is the empty diagram.
\end{center}
and the two compositions: spacial ($.\otimes.$) and sequential ($.\circ.$).

The standard interpretation is defined as:
\titlerule{$\interp{.}$}
$$ \interp{D_1\otimes D_2}:=\interp{D_1}\otimes\interp{D_2} \qquad \qquad
\interp{D_2\circ D_1}:=\interp{D_2}\circ\interp{D_1}$$
$$\interp{\input{./figures/empty-diagram.tikz}~}:=\ket{}
\qquad\qquad
\interp{~\begin{tikzpicture}
	\begin{pgfonlayer}{nodelayer}
		\node [style=none] (0) at (0, 0.2499999) {};
		\node [style=none] (1) at (0, -0.2499999) {};
	\end{pgfonlayer}
	\begin{pgfonlayer}{edgelayer}
		\draw (0.center) to (1.center);
	\end{pgfonlayer}
\end{tikzpicture}~~}:= 
\begin{pmatrix}1 & 0 \\ 0 & 1\end{pmatrix}
$$
$$ 
\interp{\input{./figures/crossing.tikz}}:= 
\begin{pmatrix}
1&0&0&0\\
0&0&1&0\\
0&1&0&0\\
0&0&0&1
\end{pmatrix} 
\qquad\qquad
\interp{\input{./figures/zw-cross.tikz}}:= 
\begin{pmatrix}
1&0&0&0\\
0&0&1&0\\
0&1&0&0\\
0&0&0&-1
\end{pmatrix}
$$
$$
\interp{\raisebox{-0.4em}{$\begin{tikzpicture}
	\begin{pgfonlayer}{nodelayer}
		\node [style=none] (a0) at (-0.2500001, -0) {};
		\node [style=none] (a1) at (0.2500001, -0) {};
		\node [style=none] (a2) at (0, 0.25) {};
	\end{pgfonlayer}
	\begin{pgfonlayer}{edgelayer}
		\draw [bend left=90, looseness=1.75] (a0.center) to (a1.center);
	\end{pgfonlayer}
\end{tikzpicture}$}}:= 
\begin{pmatrix}
1\\0\\0\\1
\end{pmatrix}
\qquad\qquad
\interp{\raisebox{-0.3em}{$\begin{tikzpicture}
	\begin{pgfonlayer}{nodelayer}
		\node [style=none] (0) at (-0.2500001, 0.2500001) {};
		\node [style=none] (1) at (0.2500001, 0.2500001) {};
	\end{pgfonlayer}
	\begin{pgfonlayer}{edgelayer}
		\draw [bend right=90, looseness=1.75] (0.center) to (1.center);
	\end{pgfonlayer}
\end{tikzpicture}$}}:= 
\begin{pmatrix}
1&0&0&1
\end{pmatrix}
$$
$$
\interp{\input{./figures/ZW-gen-GHZ-r.tikz}} = 
\begin{pmatrix}
1&0&\cdots&0&0\\
0&0&\cdots&0&0\\
\vdots & \vdots & \ddots &\vdots & \vdots\\
0&0&\cdots&0&0\\
0&0&\cdots&0&r
\end{pmatrix}
$$
\begin{minipage}{\columnwidth}
	$$
	\interp{\begin{tikzpicture}
	\begin{pgfonlayer}{nodelayer}
		\node [style=dot] (0) at (0, -0) {};
		\node [style=none] (1) at (0, 0.5) {};
		\node [style=none] (2) at (0, -0.5) {};
	\end{pgfonlayer}
	\begin{pgfonlayer}{edgelayer}
		\draw (1.center) to (2.center);
	\end{pgfonlayer}
\end{tikzpicture}}:= \begin{pmatrix}
	0&1\\1&0
	\end{pmatrix} \qquad \qquad
	\interp{\begin{tikzpicture}
	\begin{pgfonlayer}{nodelayer}
		\node [style=dot] (0) at (0, 0) {};
		\node [style=none] (1) at (0, 0.5000001) {};
		\node [style=none] (2) at (-0.25, -0.5) {};
		\node [style=none] (3) at (0.25, -0.5) {};
	\end{pgfonlayer}
	\begin{pgfonlayer}{edgelayer}
		\draw (0) to (2.center);
		\draw (3.center) to (0);
		\draw (0) to (1.center);
	\end{pgfonlayer}
\end{tikzpicture}}:= \begin{pmatrix}
	0&1\\1&0\\1&0\\0&0
	\end{pmatrix}
	$$
	\rule{\columnwidth}{0.5pt}
\end{minipage}~\\

\begin{figure*}[!htb]
	\def\scale{0.7}
	\centering
	\begin{tabular}{|ccc|}
		\hline
		&&\\
		\scalebox{\scale}{\input{./figures/ZW-rule-0-no-braid.tikz}} &  & \scalebox{\scale}{\input{./figures/ZW-rule-1.tikz}} \\
		&&\\
		\scalebox{\scale}{\input{./figures/ZW-rule-2-no-braid.tikz}} && \scalebox{\scale}{\input{./figures/ZW-rule-3.tikz}} \\
		&&\\
		\scalebox{\scale}{\input{./figures/ZW-rule-4.tikz}} && \scalebox{\scale}{\input{./figures/ZW-rule-5-no-braid.tikz}} \\
		&&\\
		\multicolumn{3}{|c|}{\scalebox{\scale}{\input{./figures/ZW-rule-6-no-braid.tikz}} $\qquad$ \scalebox{\scale}{\input{./figures/ZW-rule-X-no-braid.tikz}}}\\
		&&\\
		\multicolumn{3}{|c|}{\begin{tabular}{ccc}
				\scalebox{\scale}{\input{./figures/ZW-rule-7-no-braid.tikz}} &$\qquad$& \scalebox{\scale}{\input{./figures/reidmeister-3.tikz}}
		\end{tabular}} \\
		&&\\
		\hline
	\end{tabular}
	\caption{Set of rules for the ZW-Calculus. $r,s\in\mathbb{C}$.}
	\label{fig:ZW_rules}
\end{figure*}

\subsection{ZH-calculus}
\label{calc:ZH}

\tikzstyle{gn}=[zxnode ,fill=white]
\tikzstyle{rn}=[zxnode ,fill=gray]
\tikzstyle{H box}=[rectangle,fill=white,draw=black,xscale=1,yscale=1,font=\footnotesize,inner sep=1.2pt,minimum width=0.15cm,minimum height=0.15cm]

The ZH-diagrams are generated by:
\begin{center}
	\bgroup
	\def\arraystretch{2.5}
	{\begin{tabular}{|cc|cc|}
			\hline
			$Z^{(n,m)}:n\to m$ & \vphantom{\scalebox{1.15}{\input{./figures/gn-0.tikz}}}\input{./figures/gn-0.tikz} & $\mathbb{I}:1\to 1$ & \begin{tikzpicture}
	\begin{pgfonlayer}{nodelayer}
		\node [style=none] (0) at (0, 0.2499999) {};
		\node [style=none] (1) at (0, -0.2499999) {};
	\end{pgfonlayer}
	\begin{pgfonlayer}{edgelayer}
		\draw (0.center) to (1.center);
	\end{pgfonlayer}
\end{tikzpicture}\\\hline
			$X^{(n,m)}:n\to m$ & \vphantom{\scalebox{1.15}{\input{./figures/rn-0.tikz}}}\input{./figures/rn-0.tikz} & $e:0\to 0$ & \input{./figures/empty-diagram.tikz}\\\hline
			$H^{(n,m)}(a):n\to m$ & \vphantom{\scalebox{1.15}{\input{./figures/ZH-gen-H-a.tikz}}}\input{./figures/ZH-gen-H-a.tikz} & $\neg:1\to1$ & \begin{tikzpicture}
	\begin{pgfonlayer}{nodelayer}
		\node [style=none] (0) at (0, 0.5) {};
		\node [style=rn] (1) at (0, -0) {$\neg$};
		\node [style=none] (2) at (0, -0.5) {};
	\end{pgfonlayer}
	\begin{pgfonlayer}{edgelayer}
		\draw (2.center) to (0.center);
	\end{pgfonlayer}
\end{tikzpicture} \\\hline
			$\sigma:2\to 2$ & \input{./figures/crossing.tikz} & $\epsilon:2\to 0$ & \begin{tikzpicture}
	\begin{pgfonlayer}{nodelayer}
		\node [style=none] (0) at (-0.2500001, 0.2500001) {};
		\node [style=none] (1) at (0.2500001, 0.2500001) {};
	\end{pgfonlayer}
	\begin{pgfonlayer}{edgelayer}
		\draw [bend right=90, looseness=1.75] (0.center) to (1.center);
	\end{pgfonlayer}
\end{tikzpicture}\\\hline
			$\eta:0\to 2$ & \begin{tikzpicture}
	\begin{pgfonlayer}{nodelayer}
		\node [style=none] (a0) at (-0.2500001, -0) {};
		\node [style=none] (a1) at (0.2500001, -0) {};
		\node [style=none] (a2) at (0, 0.25) {};
	\end{pgfonlayer}
	\begin{pgfonlayer}{edgelayer}
		\draw [bend left=90, looseness=1.75] (a0.center) to (a1.center);
	\end{pgfonlayer}
\end{tikzpicture}&\multicolumn{2}{c}{}\\\cline{1-2}
	\end{tabular}}
	\egroup\\
	where $n,m\in \mathbb{N}$, $a \in \mathbb{C}$, and the generator $e$ is the empty diagram.
\end{center}
and the two compositions: spacial ($.\otimes.$) and sequential ($.\circ.$).

The language was introduced to allow a simple representation of hypergraph states and multi-controlled-Z gates. To do so it features a node called $H$-spider, which can be seen as a generalisation of the Hadamard gate. By convention, when no parameter is specified in $H$, the implicit parameter taken is $-1$: \input{./figures/ZH-convention.tikz}.

\begin{figure*}[!htb]
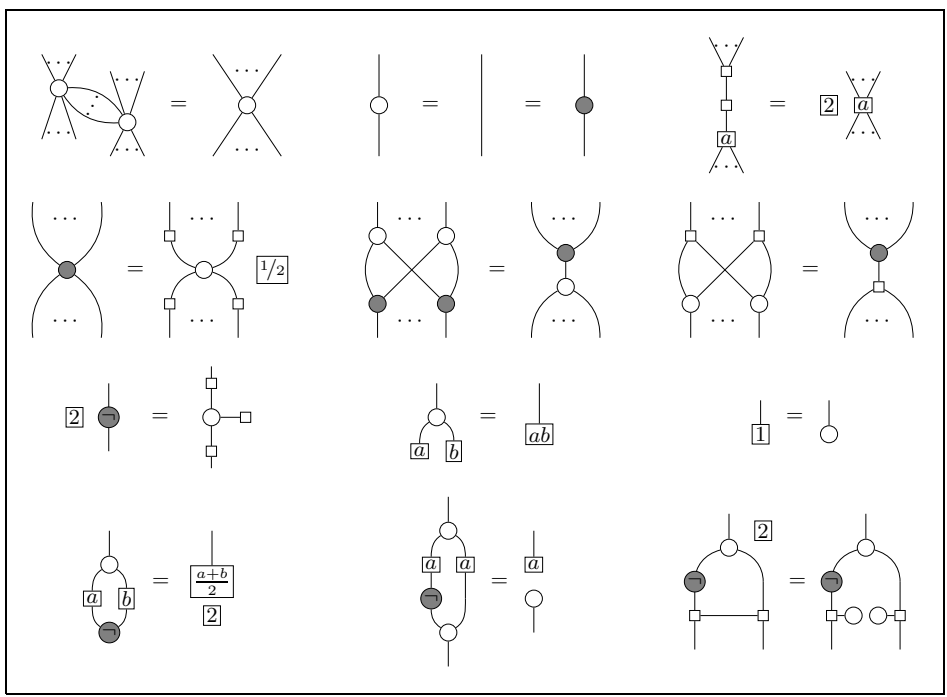

	\centering
	\hypertarget{r:ZH}{}
	\scalebox{0.9}{
	\begin{tabular}{|@{~~~}c@{$\qquad\quad$}c@{$\quad\qquad$}c@{~~~}|}
		\hline
		&& \\
		\input{./figures/ZH-rule-ZS1.tikz}&
		\input{./figures/ZH-rule-ZS2.tikz}&
		\input{./figures/ZH-rule-HS1.tikz}\\
		&& \\
		\input{./figures/ZH-rule-red.tikz}&
		\input{./figures/ZH-rule-BA1.tikz}&
		\input{./figures/ZH-rule-BA2.tikz}\\
		&& \\
		\input{./figures/ZH-rule-N.tikz}&
		\input{./figures/ZH-rule-M.tikz}&
		\input{./figures/ZH-rule-U.tikz}\\
		&& \\
		\input{./figures/ZH-rule-A.tikz}&
		\input{./figures/ZH-rule-I.tikz}&
		\input{./figures/ZH-rule-O.tikz}\\
		&& \\
		\hline
	\end{tabular}}
	\caption[]{Set of rules ZH. (...) denote zero or more wires, while (\protect\rotatebox{45}{\raisebox{-0.4em}{$\cdots$}}) denote one or more wires.
	}
	\label{fig:ZH_rules}
\end{figure*}

The language comes with a standard interpretation defined as:
\titlerule{$\interp{.}$}
$$ \interp{D_1\otimes D_2}:=\interp{D_1}\otimes\interp{D_2} \qquad \qquad
\interp{D_2\circ D_1}:=\interp{D_2}\circ\interp{D_1}$$
$$\interp{\input{./figures/empty-diagram.tikz}~}:=\begin{pmatrix}
1
\end{pmatrix} \qquad\qquad
\interp{~\begin{tikzpicture}
	\begin{pgfonlayer}{nodelayer}
		\node [style=none] (0) at (0, 0.2499999) {};
		\node [style=none] (1) at (0, -0.2499999) {};
	\end{pgfonlayer}
	\begin{pgfonlayer}{edgelayer}
		\draw (0.center) to (1.center);
	\end{pgfonlayer}
\end{tikzpicture}~~}:= \begin{pmatrix}
1 & 0 \\ 0 & 1\end{pmatrix}
\qquad\qquad
\interp{\begin{tikzpicture}
	\begin{pgfonlayer}{nodelayer}
		\node [style=none] (0) at (0, 0.5) {};
		\node [style=rn] (1) at (0, -0) {$\neg$};
		\node [style=none] (2) at (0, -0.5) {};
	\end{pgfonlayer}
	\begin{pgfonlayer}{edgelayer}
		\draw (2.center) to (0.center);
	\end{pgfonlayer}
\end{tikzpicture}}:=\begin{pmatrix}
0& 1 \\ 1 & 0
\end{pmatrix}$$
$$
\interp{\input{./figures/crossing.tikz}}:= \begin{pmatrix}
1&0&0&0\\
0&0&1&0\\
0&1&0&0\\
0&0&0&1
\end{pmatrix} \qquad\qquad
\interp{\raisebox{-0.35em}{$\begin{tikzpicture}
	\begin{pgfonlayer}{nodelayer}
		\node [style=none] (a0) at (-0.2500001, -0) {};
		\node [style=none] (a1) at (0.2500001, -0) {};
		\node [style=none] (a2) at (0, 0.25) {};
	\end{pgfonlayer}
	\begin{pgfonlayer}{edgelayer}
		\draw [bend left=90, looseness=1.75] (a0.center) to (a1.center);
	\end{pgfonlayer}
\end{tikzpicture}$}}:= \begin{pmatrix}
1\\0\\0\\1
\end{pmatrix}$$
$$\interp{\raisebox{-0.25em}{$\begin{tikzpicture}
	\begin{pgfonlayer}{nodelayer}
		\node [style=none] (0) at (-0.2500001, 0.2500001) {};
		\node [style=none] (1) at (0.2500001, 0.2500001) {};
	\end{pgfonlayer}
	\begin{pgfonlayer}{edgelayer}
		\draw [bend right=90, looseness=1.75] (0.center) to (1.center);
	\end{pgfonlayer}
\end{tikzpicture}$}}:= \begin{pmatrix}
1&0&0&1
\end{pmatrix}
\qquad\qquad
\interp{\input{./figures/ZH-gen-H-a.tikz}}:= \begin{pmatrix}
1 & \cdots & 1 & 1 \\
\vdots & \ddots & \vdots & \vdots \\
1 & \cdots & 1 & 1 \\
1 & \cdots & 1 & a
\end{pmatrix}$$
$$
\interp{\input{./figures/gn-0.tikz}}:=
\begin{pmatrix}
1 & 0 & \cdots & 0 & 0 \\
0 & 0 & \cdots & 0 & 0 \\
\vdots & \vdots & \ddots & \vdots & \vdots \\
0 & 0 & \cdots & 0 & 0 \\
0 & 0 & \cdots & 0 & 1
\end{pmatrix}$$
\begin{minipage}{\columnwidth}
	$$\interp{\input{./figures/rn-0.tikz}}:=\frac{1}{2}\interp{~\begin{tikzpicture}
	\begin{pgfonlayer}{nodelayer}
		\node [style=H box] (0) at (0, 0) {};
		\node [style=none] (1) at (0, 0.275) {};
		\node [style=none] (2) at (0, -0.275) {};
	\end{pgfonlayer}
	\begin{pgfonlayer}{edgelayer}
		\draw (2.center) to (1.center);
	\end{pgfonlayer}
\end{tikzpicture}
~}^{\otimes m}\circ \interp{\input{./figures/gn-0.tikz}}\circ \interp{~\begin{tikzpicture}
	\begin{pgfonlayer}{nodelayer}
		\node [style=H box] (0) at (0, 0) {};
		\node [style=none] (1) at (0, 0.275) {};
		\node [style=none] (2) at (0, -0.275) {};
	\end{pgfonlayer}
	\begin{pgfonlayer}{edgelayer}
		\draw (2.center) to (1.center);
	\end{pgfonlayer}
\end{tikzpicture}
~}^{\otimes n}
	$$
	\rule{\columnwidth}{0.5pt}
\end{minipage}\\

A set of rules was proposed together with the language (Figure \ref{fig:ZH_rules}). It makes the ZH-Calculus complete for $\mathbf{Qubit}$.

\section{ZX Axiomatisations}

\tikzstyle{gn}=[zxnode ,fill=green]
\tikzstyle{rn}=[zxnode ,fill=red]
\tikzstyle{H box}=[rectangle,fill=yellow,draw=black,xscale=1,yscale=1,font=\footnotesize,inner sep=1.2pt,minimum width=0.15cm,minimum height=0.15cm]

\begin{figure*}[!htb]
	\centering
	\hypertarget{r:rules-c}{}
	\resizebox{\columnwidth}{!}{
	\begin{tabular}{|@{}c@{}|}
		\hline
		\\
		\begin{tabular}{c@{$\qquad\qquad$}c@{$\qquad\qquad$}c}
			\input{./figures/spider-1.tikz}&\input{./figures/s2-simple.tikz}&\\\\
			\input{./figures/b1s.tikz}&\input{./figures/b2s.tikz}&\input{./figures/bicolor_pi_4_eq_empty.tikz}\\\\
			\input{./figures/k2s.tikz}&\input{./figures/euler-decomp-scalar-free.tikz}&\input{./figures/h2.tikz}
		\end{tabular}\\\\
		\begin{tabular}{c@{$\qquad\qquad$}c}
			\input{./figures/former-supp.tikz}&\input{./figures/commutation-of-controls-general-simplified.tikz} \\\\
			\input{./figures/BW-simplified.tikz}&\input{./figures/add-axiom-3.tikz}
		\end{tabular}\\\\
		\hline
	\end{tabular}}
	\caption[]{Set of rules for the general \zx-Calculus with scalars. All of these rules also hold when flipped upside-down, or with the colours red and green swapped. The right-hand side of (E) is an empty diagram. (...) denote zero or more wires, while (\protect\rotatebox{45}{\raisebox{-0.5em}{$\cdots$}}) denote one or more wires.}
	\label{fig:ZX_rules}
\end{figure*}

%

%

\begin{figure*}[!htb]
	\centering
	\hypertarget{r:rules-clif}{}
	\scalebox{0.9}{
	\begin{tabular}{|@{$\quad$}c@{$\quad$}|}
		\hline
		\\
		\input{./figures/spider-1.tikz}$\qquad\qquad$\input{./figures/s2-simple.tikz}$\qquad\qquad$\input{./figures/inverse-no-label.tikz}\\\\
		\input{./figures/b1s.tikz}$\qquad\qquad$\input{./figures/b2s.tikz}\\\\
		\input{./figures/euler-decomp-scalar-free.tikz}$\qquad\qquad$\input{./figures/h2.tikz}$\qquad\qquad$\input{./figures/zero-rule.tikz}\\\\
		\hline
	\end{tabular}}
	\caption[]{Set of rules $\zx_{\frac{\pi}{2}}$ for the \frag2 of the \zx-Calculus with scalars. All of these rules also hold when flipped upside-down, or with the colours red and green swapped. The right-hand side of (IV) is an empty diagram. (...) denote zero or more wires, while (\protect\rotatebox{45}{\raisebox{-0.5em}{$\cdots$}}) denote one or more wires.}
	\label{fig:ZX_rules_clifford}
\end{figure*}

\end{document}